\documentclass[conference]{IEEEtran}

\ifCLASSINFOpdf
\else
\fi
\usepackage{cite}
\usepackage{graphicx}
\usepackage{textcomp}

\usepackage{xspace}
\newcommand{\eg}{\emph{e.g.,}\xspace}
\newcommand{\ie}{\emph{i.e.,}\xspace}

\newcommand{\first}{(i)\xspace}
\newcommand{\second}{(ii)\xspace}
\newcommand{\third}{(iii)\xspace}
\newcommand{\fourth}{(iv)\xspace}

\newcommand\revised[1]{{\color{orange} #1}}
\renewcommand\revised[1]{{#1}}
\newcommand\morerevised[1]{{\color{orange} #1}}
\renewcommand\morerevised[1]{{#1}}

\newcommand{\ours}{\textsc{Shaman}\xspace}

\newcommand{\oursmotto}{\emph{``In ancient times, a shaman guided through unseen perils.''}}

\usepackage[dvipsnames]{xcolor}
\usepackage[hidelinks]{hyperref}

\usepackage{amsmath,amsfonts,amssymb}
\usepackage{mathtools}
\usepackage{bm}
\usepackage{amsthm}
\theoremstyle{definition}
\newtheorem{theorem}{Theorem}
\newtheorem{definition}{Definition}
\newtheorem{lemma}{Lemma}
\usepackage[skip=6pt]{subcaption}
\usepackage{tabu}
\usepackage{threeparttablex}
\usepackage[export]{adjustbox}
\usepackage{booktabs}
\usepackage{multirow}
\usepackage{multicol}
\usepackage{algorithm}
\usepackage[noend]{algpseudocode}
\usepackage{balance}

\usepackage{enumitem}
\usepackage[T1]{fontenc}
\usepackage[utf8]{inputenc}
\usepackage{mathptmx}
\usepackage{pifont}

\usepackage{setspace}
\usepackage{colortbl}
\usepackage{makecell}
\usepackage{diagbox}
\usepackage{arydshln}
\usepackage{tikz}
\usepackage{afterpage}

\newcommand\mypart[1]{\textsl{\footnotesize #1}}

\newcommand*\circled[1]{\tikz[baseline=(char.base)]{
            \node[shape=circle,fill,inner sep=.5pt] (char) {\textcolor{white}{#1}};}}

\usepackage{etoolbox}

\usepackage[available,functional,reproduced]{ndssbadges}

\begin{document}
%
\title{Understanding the Stealthy BGP Hijacking\\Risk in the ROV Era}

\author{
    \IEEEauthorblockN{
        Yihao Chen\IEEEauthorrefmark{4}\IEEEauthorrefmark{2},
        Qi Li\IEEEauthorrefmark{3}\IEEEauthorrefmark{2},
        Ke Xu\IEEEauthorrefmark{5}\IEEEauthorrefmark{2},
        Zhuotao Liu\IEEEauthorrefmark{3}\IEEEauthorrefmark{2},
        Jianping Wu\IEEEauthorrefmark{3}\IEEEauthorrefmark{2},
    }
    \IEEEauthorblockA{
        \IEEEauthorrefmark{4}DCST \& BNRist \& State Key Laboratory of Internet Architecture, Tsinghua University
    }
    \IEEEauthorblockA{
        \IEEEauthorrefmark{3}INSC \& State Key Laboratory of Internet Architecture, Tsinghua University,
        \IEEEauthorrefmark{2}Zhongguancun Laboratory
    }
    \IEEEauthorblockA{
        \IEEEauthorrefmark{5}DCST \& State Key Laboratory of Internet Architecture, Tsinghua University
    }
    \IEEEauthorblockA{
        yh-chen21@mails.tsinghua.edu.cn,
        \{qli01, xuke, zhuotaoliu\}@tsinghua.edu.cn,
        jianping@cernet.edu.cn
    }
}
	

%


\IEEEoverridecommandlockouts
\makeatletter\def\@IEEEpubidpullup{6.5\baselineskip}\makeatother
\IEEEpubid{\parbox{\columnwidth}{
\vspace{-3.96\baselineskip}
		Network and Distributed System Security (NDSS) Symposium 2026\\
		23 - 27 February 2026 , San Diego, CA, USA\\
		ISBN 979-8-9919276-8-0\\  
		https://dx.doi.org/10.14722/ndss.2026.230097\\
		www.ndss-symposium.org
}
\hspace{.8\columnsep}\makebox[\columnwidth]{}}

\maketitle

\begin{abstract}
The partial deployment of Route Origin Validation (ROV) poses an unexpected security threat known as stealthy BGP hijacking, \ie a particularly elusive form of BGP hijacking where malicious routes divert traffic without reaching (and thus alerting) the victims. This risk remains largely unexplored, with neither documented real-world incidents nor systematic characterization available. To bridge this gap, we formalize stealthy BGP hijacking and propose heuristics to discover potential instances through routing table discrepancies. We conduct the first \emph{empirical} study to track and profile stealthy BGP hijacking in the wild, contributing a curated real-world incident dataset and a long-term monitoring service. Inspired by the empirical insights, we further conduct an \emph{analytical} study to exhaustively assess the risk. This requires accurate ROV deployment data, complete Internet-wide routes, and tailored analytical models. To address these challenges, we develop \ours, a BGP route inference framework dedicated to assessing stealthy BGP hijacking risk. \ours consolidates multiple sources to construct an accurate view of ROV deployment, infers complete Internet-wide routes through a highly efficient matrix-based approach, and facilitates statistical risk analysis via a ``victim-target-hijacker'' 3-tuple model. By reducing the time for generating Internet-scale routes from over three months to just 5.22 hours, \ours enables systematic risk assessment across 8.3 billion generated routes under real-world ROV deployment. Our findings reveal a 14.1\% overall success probability for stealthy BGP hijacking, with targeted attacks reaching 99.5\% success in specific cases. Validation against our real-world dataset shows up to 95.9\% incident-level accuracy, demonstrating the fidelity of our analytical results.
\end{abstract}


%
\IEEEpeerreviewmaketitle

\section{Introduction}

The Border Gateway Protocol (BGP) has been known for its security vulnerabilities, with the infamous BGP hijacking being a major threat. To counter this threat, the community has proposed various security enhancements~\cite{lepinski2017bgpsec,oorschot2007interdomain,white2003securing,kent2000secure,mohapatra2013bgp,lepinski2012rfc,bush2014rfc}, among which the Resource Public Key Infrastructure (RPKI) and Route Origin Validation (ROV) show the promise in real-world deployment. ROV-enabled Autonomous Systems (ASes) can retrieve authorized registries managed by RPKI, known as Route Origin Authorizations (ROAs), based on which they can verify the correctness of prefix-origin associations contained in received BGP announcements. As of March 2025, ROAs have covered 57.1\% of globally routable IPv4 prefixes~\cite{nist}, yet the deployment of ROV is relatively limited, with only hundreds to thousands of ASes identified as ROV-enabled~\cite{li2023rovista,cloudflare,apnic}.

Presumably, ROV will remain in partial deployment for a relatively long time, which, in addition to offering incomplete protection as a consequence, results in an unexpected security threat, \ie \emph{highly stealthy BGP hijacking that is effectively invisible from the victim on the control plane}. This new threat, which we refer to as ROV-related stealthy BGP hijacking, or for short, \emph{stealthy hijacking}, occurs when an AS, despite being nominally protected by ROV-enabled ASes from receiving malicious routes, has its traffic silently diverted to a hijacker through legacy ASes along the data plane path. It is particularly insidious because the affected AS remains unaware of malicious routes throughout the hijacking, rendering common control-plane based protections ineffective in practice.

This highlights the unexpected downside of partial ROV deployment, yet the issue remains largely unexplored. No real-world stealthy hijacking incidents have been documented, and a systematic investigation into its prevalence and impacts is still missing. A recent study~\cite{morillo2021rov++} takes a pioneering step towards mitigating stealthy hijacking via proactive rerouting and blackholing. Yet, its mitigation-oriented focus provides limited real-world evidence or heuristics for tracking and profiling the threat (see further  discussions in \S\ref{sec:related-work}).

To bridge this gap, we seek insights from real-world observations. However, the lack of an established definition of stealthy hijacking makes it difficult to identify the threat. To this end, we formalize stealthy BGP hijacking and derive heuristics to discover hijacking instances based on routing table discrepancies observed across vantage points. Our rationale behind is to determine if any AS along legitimate routes can forward traffic to potential hijackers. Using routing tables from RouteViews vantage points, we conduct the first empirical study to track and profile stealthy hijacking in the wild (\S\ref{sec:empirical-study}). We capture 1,393 potential incidents over a two-month window in 2025, and analyze their impacts and causes extensively. We further validate these observations against a broad knowledge base including RPKI, IRR and WHOIS, which results in a curated dataset of 318 high-confidence incidents covering 2,178 routes. This dataset, along with our long-term monitoring service continuing to report real-world incidents, is publicly available at \url{https://yhchen.cn/stealthy-bgp-hijacking}.

A key observation from our empirical study is that the visibility of stealthy hijacking is sensitive to vantage point selection and subject to continual change. This inspires us to further propose comprehensive and deterministic risk assessment through an analytical approach, which, however, faces three key challenges. First, determining the current state of ROV deployment is nontrivial. Unlike ROAs, ROV deployment is not publicly disclosed by default, and existing ROV measurements remain limited in coverage. Second, given the complexity of Internet topology, acquiring complete knowledge of Internet-wide routes to fully assess the risk is challenging. Third, no dedicated analytical model currently exists to characterize stealthy hijacking risk at a fine-grained level, yet such a model is essential for systematic risk assessment.

To address these challenges, we develop \ours, a BGP route inference framework dedicated to assessing stealthy hijacking risk (\S\ref{sec:analytical-approach}). \ours consolidates ROV measurements from multiple sources to ensure an accurate view of current ROV deployment. It utilizes a matrix-based approach to infer complete knowledge of legitimate AS-level routes and potential AS-to-hijacker routes under partial ROV deployment. This approach encodes essential routing information in compact matrices through a unique one-byte encoding scheme, and leverages highly optimized matrix operations for efficient route inference. Combined with a topology compression method, \ours generates routes across \emph{all} ASes within hours, reducing the otherwise months-long runtime of prior art~\cite{brandt2021optimized,furuness2023bgpy,li2023realizing}. It further applies a ``victim-target-hijacker'' three-tuple model to statistically characterize fine-grained hijacking instances, thereby supporting systematic role-based risk analysis.

Through \ours, we conduct comprehensive analytical risk assessment with the current Internet topology, examining 8.3 billion Internet-wide routes across 77,600 ASes (\S\ref{sec:risk-assessment}). Our analysis exhausts ``victim-target-hijacker'' instances that are vulnerable to stealthy hijacking, and characterizes their prevalence, distribution, and topological features, yielding seven key insights. We reveal that the current partial ROV deployment introduces a 14.1\% overall success probability for stealthy hijacking, with targeted attacks reaching 99.5\% success in specific cases. We also evaluate \ours extensively (\S\ref{sec:performance-evaluation}). We show that it reduces the time for Internet-scale route generation from over three months to just 5.22 hours, achieving a 500-fold speedup over existing methods. Validation against our curated dataset demonstrates up to 95.9\% incident-level accuracy, confirming the fidelity of our analytical results. Lastly, we conduct ablation experiments to justify \ours's design choice of consolidating multiple ROV measurement sources, and validate its robustness against input noise.

To summarize, our contributions are four-fold:
\begin{itemize}[nosep]
\item We develop effective heuristics to discover stealthy BGP hijacking instances based on routing table discrepancies.
\item We conduct the first empirical study to track stealthy BGP hijacking in the wild, establishing a curated real-world incident dataset and a long-term monitoring service.
\item Motivated by empirical insights, we design \ours, a framework dedicated to systematic assessment of stealthy BGP hijacking risk, and evaluate it extensively.
\item Through \ours, we assess stealthy BGP hijacking risk in the current Internet thoroughly, deriving seven key insights while achieving 95.9\% incident-level accuracy.
\end{itemize}
\section{Background}
\label{sec:background}

\noindent\textbf{Interdomain Routing and BGP Hijacking.}
As of March 2025, the Internet comprises over 77,600 Autonomous Systems (ASes) in interdomain routing~\cite{potaroo2021}, each identified by a unique AS Number (ASN). These ASes interconnect via the Border Gateway Protocol (BGP), where each AS announces its IP prefixes to neighbors through route announcements. Upon receiving an announcement, an AS extracts the AS-level path destined for a specified prefix, updates its Routing Information Base (RIB), and performs best-route selection. The chosen route is then appended with the AS's own ASN and further propagated to selected neighbors. Both best-route selection and propagation are influenced by the AS's routing policies, \eg predefined route preferences based on business relationships.

BGP lacks native security mechanisms and is vulnerable to misinformation injected by malicious actors. A consequent threat is BGP hijacking, where a malicious AS falsely claims to originate a prefix it does not own. If other ASes accept the forged announcement, they reroute traffic accordingly and divert it to the hijacker. BGP hijacking takes two forms: \emph{exact-prefix hijacking}, where the hijacker announces the exact prefix owned by the legitimate origin, and \emph{sub-prefix hijacking}, where a more specific sub-prefix is announced. The latter is typically more damaging, since BGP routers prioritize the longest-prefix match when forwarding traffic. Real-world BGP hijacking incidents have reported serious consequences, \eg redirecting cryptocurrency funds to attacker-controlled accounts~\cite{kentik}.

\noindent\textbf{RPKI and ROV for BGP Security.}
To counter BGP hijacking, the community proposed various security extensions~\cite{lepinski2017bgpsec,oorschot2007interdomain,kent2000secure}. Among them, Resource Public Key Infrastructure (RPKI)~\cite{lepinski2012rfc} and Route Origin Validation (ROV)~\cite{bush2014rfc} have gained significant real-world adoption. RPKI provides a cryptographic framework for prefix holders to publish Route Origin Authorizations (ROAs), which specify valid origin ASes and maximum allowable prefix lengths. Meanwhile, ROV refers to the operational practice by which ASes retrieve and verify these ROAs to validate BGP announcements, thereby ensuring that  the announcing AS is authorized to advertise a specific prefix. Invalid BGP announcements are typically discarded.

RPKI has gained great traction in recent years, but the actual state of ROV deployment remains nontrivial to know. Existing measurements of ROV deployment are mostly best-effort and limited in coverage, estimating anywhere from several hundred to over three thousand ROV-enabled ASes across the Internet~\cite{li2023rovista,chen2022rov,hlavacek2023keep,cloudflare,apnic}. This incomplete deployment greatly throttles ROV's effectiveness, \eg even with 60\% global ROV adoption, sub-prefix hijacking still succeeds 40\% of the time~\cite{morillo2021rov++}. More critically, partial ROV deployment introduces an unintended security risk, as we highlight next.
\begin{figure*}[t]
\captionsetup[sub]{belowskip=0pt}
\captionsetup{belowskip=-2mm}
\centering
    \subcaptionbox{Impact of partial ROV deployment on BGP hijacking.\label{fig:problem-statement-1}}{\includegraphics[width=.45\linewidth]{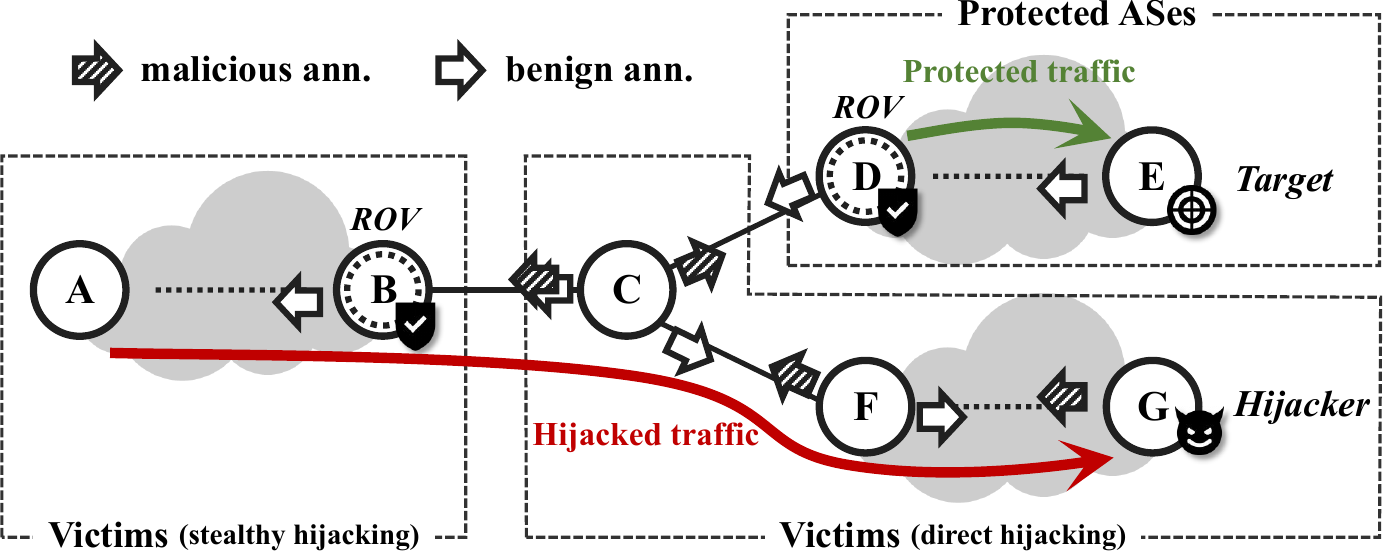}}
    \hspace{.05\linewidth}
    \subcaptionbox{Expected traffic propagation from AS A's view.\label{fig:problem-statement-2}}{\includegraphics[width=.45\linewidth]{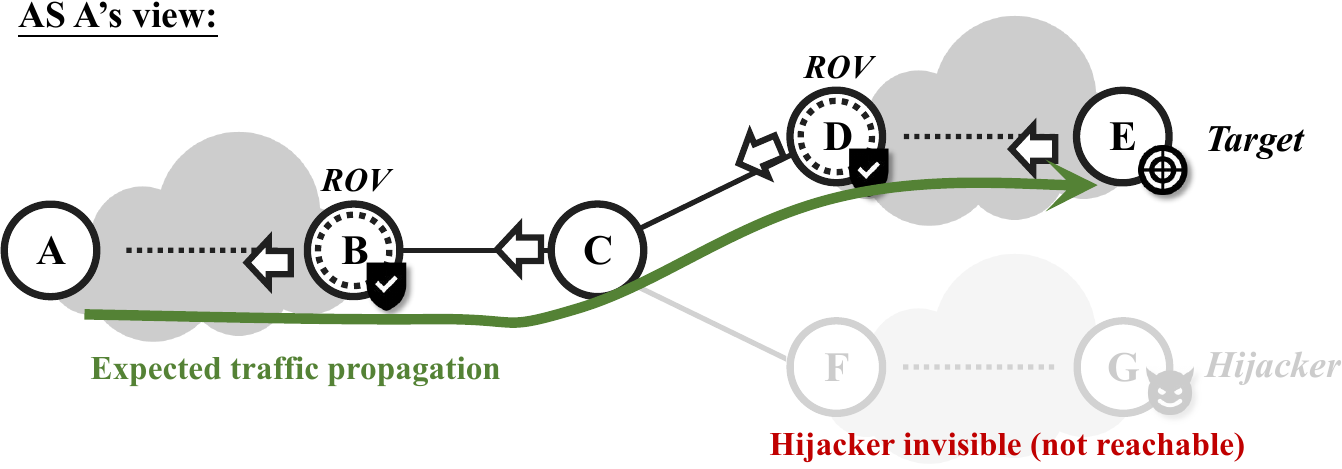}}
    \par\bigskip
    \subcaptionbox{Actual traffic propagation under sub-prefix hijacking.\label{fig:problem-statement-3}}{\includegraphics[width=.45\linewidth]{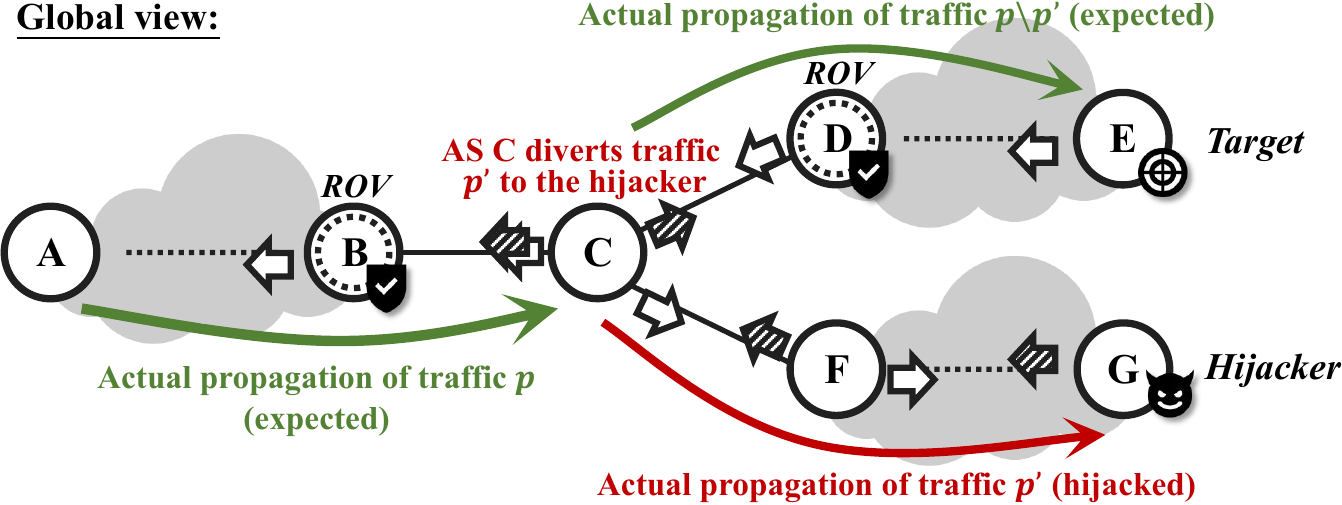}}
    \hspace{.05\linewidth}
    \subcaptionbox{Actual traffic propagation under exact-prefix hijacking.\label{fig:problem-statement-4}}{\includegraphics[width=.45\linewidth]{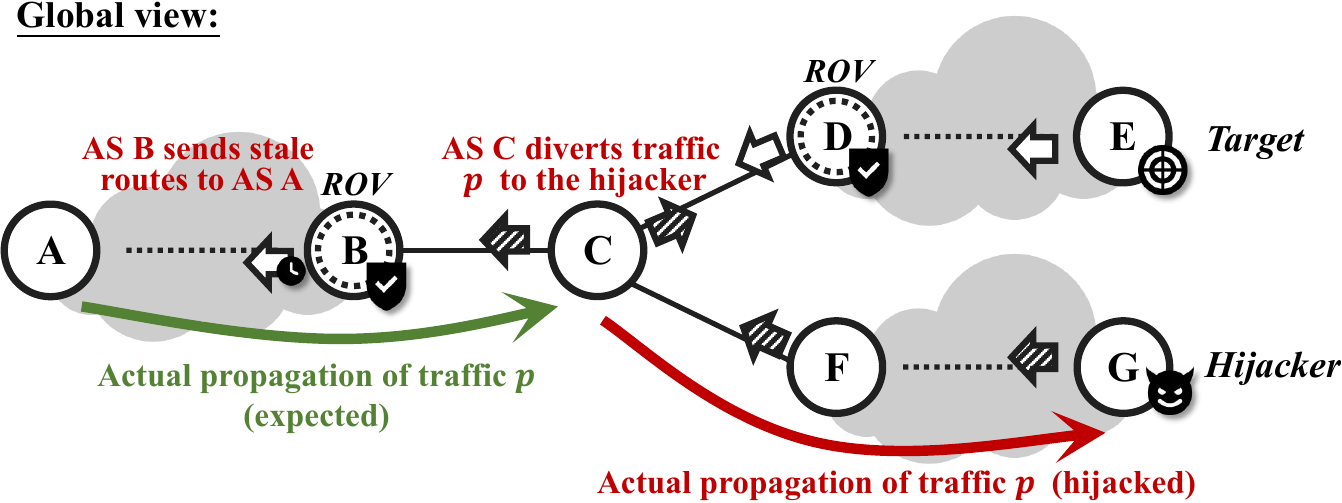}}
\caption{Stealthy hijacking under partial ROV deployment.}
\label{fig:problem-statement}
\end{figure*}

\section{Problem Statement}
\label{sec:problem-statement}

In this section, we articulate the unintended security risk of partial ROV deployment and define the relevant concepts.

\subsection{The Unexpected Downside of ROV}
\label{sec:unexpected-downside-of-rov}

ROV is proposed to prevent BGP hijacking. However, in case of partial deployment, its effectiveness can be undermined. Figure~\ref{fig:problem-statement} illustrates a scenario where BGP hijacking succeeds despite ROV deployment. We abstract the Internet topology as a graph, where vertices represent ASes and edges represent interdomain links. In this scenario, the \emph{hijacker} (AS G) launches an attack by announcing (sub-)prefixes owned by the \emph{target} (AS E). The ROV-enabled ASes block propagation of malicious announcements, so only AS C and F may accept the bogus route to the hijacker\footnote{For brevity, we refer to the interdomain route (accepted by X) for a specific prefix announced by Y as ``the route (from X) to Y''.} depending on routing policies.

We take AS A as an example to explain BGP hijacking in this scenario. As shown in Figure~\ref{fig:problem-statement-1}, AS A only receives routes to the target since ROV-enabled AS B correctly drops the hijacker's malicious announcement. Consequently, AS A lacks a route to the hijacker in its routing table, limiting its control-plane visibility (see Figure~\ref{fig:problem-statement-2}). AS A, unaware of the hijacker's presence, expects its traffic to reach the target correctly (green arrow). However, from a global perspective (see Figure~\ref{fig:problem-statement-3} and \ref{fig:problem-statement-4}), AS A's traffic traverses legacy AS C en route to the target. If AS C accepts the hijacker's route (\eg when the bogus route is preferred over the legitimate one or targets a sub-prefix), the traffic from AS A is actually forwarded to the hijacker. Note that AS A cannot observe the actual traffic forwarding unless it performs active data-plane probing or gains a broader view from external vantage points.

The above example showcases the unexpected downside of partial ROV deployment: 
\emph{it prevents certain ASes from observing bogus routes, contributing to the stealthiness of BGP hijacking attacks.} We refer to such BGP hijacking, which is invisible to certain victims on the control plane, as stealthy hijacking. 
In general, an AS is at risk of stealthy hijacking when it has no route to the hijacker, but at least one legacy AS along the legitimate path accepts the bogus route; this applies to both exact-prefix and sub-prefix hijacking. Based on this principle, we formally define the risk in the next section.

\subsection{Problem Formulation}

Now, we formalize stealthy hijacking risk introduced by partial ROV deployment, starting from the Internet topology:

\begin{definition}[Internet Topology]\label{def:internet-topology}
    An Internet topology is a tuple $G=(V,E,V_{rov})$ where $V$ is the set of all ASes in the Internet, $E$ is the set of links between ASes in $V$, and $V_{rov} \subseteq V$ is the set of ASes that adopt ROV.
\end{definition}

\noindent
The Internet topology exhibits partial ROV deployment if and only if $V_{rov} \notin \{V, \emptyset\}$. We then define the AS-level route:

\begin{definition}[AS-Level Route]\label{def:as-level-route}
    Given $G=(V,E,V_{rov})$, let $u,v \in V$ and let $p$ be a prefix announced by $v$. $R_G(u,v;p)$ denotes the AS-level route starting from $u$ to reach the prefix $p$ announced by $v$ after Internet routing converges on $G$ under existing routes, ROAs, and BGP policies. If it exists, $R_G(u,v;p) = (a_0, \dots, a_l)$, where $a_0=u, a_l=v, a_0, \dots, a_l \in V, (a_i, a_{i+1}) \in E \text{ for } 0 \le i \le l-1$; otherwise, $R_G(u,v;p) = ()$.
\end{definition}

\noindent
Based on them, we can define the stealthy hijacking risk:

\begin{definition}[Stealthy Hijacking Risk]\label{def:stealthy-hijacking-risk}
    Given $G=(V,E,V_{rov})$, let $u,v,w \in V$ and let $p$ be a prefix announced by $v$. The route
    $R_G(u,v;p)$ is at risk of stealthy hijacking by $w$ if there exists a prefix $p'$, such that \first $p' = p$ or $p'$ is a sub-prefix of $p$, \second $R_G(u,w;p') = ()$, \third $R_G(u,v;p) \neq ()$, and \fourth there exists $a_i \in R_G(u,v;p)$, such that $a_i \neq u, a_i \neq v$, and $R_G(a_i, w;p') \neq ()$.
\end{definition}

\noindent
Intuitively, given victim $u$, target $v$, and hijacker $w$, stealthy hijacking occurs when $u$ cannot reach $w$, yet on $u$'s route to $v$, an AS $a_i$ has a route to $w$. Thus, traffic originated from $u$ follows the route towards $v$ until it reaches $a_i$, who instead forwards the traffic to $w$. We call the ``victim-target-hijacker'' 3-tuple $(u, v, w)$ a \emph{stealthy hijacking instance}, and $a_i$ the \emph{risk-critical AS}. Each instance corresponds to a risk-critical AS; for example, in Figure~\ref{fig:problem-statement}, AS C is the risk-critical AS of the instance (A, E, G). Note that each instance only reflects the risk, meaning that a potential stealthy hijacking attack \emph{could} occur, but does not necessarily imply an actual occurrence. We base our risk analysis on the 3-tuple model, and refer to BGP hijacking that is not stealthy hijacking as \emph{direct hijacking}.

Notably, stealthy hijacking follows the established BGP hijacking attack model but manifests in a more subtle form under partial ROV deployment. Any stealthy hijacking due to partial ROV deployment would also be possible (though not necessarily stealthy) without ROV, because removing ROV only improves attacker reachability without making any benign route more preferred. Thus, regardless of ROV deployment, victims that are subject to hijacking would continue to prefer routes with risk-critical ASes. We emphasize that the shift from direct to stealthy hijacking is a byproduct of partial ROV deployment in the current Internet ecosystem, rather than a flaw in ROV itself. However, stealthy hijacking enables attacks that can evade existing control-plane defenses~\cite{chen2024learning,dong2021isp,shapira2022ap2vec} and hinder post-attack forensics~\cite{testart2019profiling}. Understanding its real-world prevalence and associated risk is thus critical to BGP security.
\begin{table}[t]

\newcommand\typeA[1]{{\cellcolor{teal!35} #1}}
\newcommand\typeB[1]{{\cellcolor{gray!20} #1}}
\newcommand\sampleA{{\protect\tikz[x=3ex,y=1.5ex] \protect\fill [teal!35] (0,0) rectangle (1,1);}\ }
\newcommand\sampleB{{\protect\tikz[x=3ex,y=1.5ex] \protect\fill [gray!20] (0,0) rectangle (1,1);}\ }
    \caption{Tags for predefined incident behaviors.}
    \label{tab:incident-tags}
    \scriptsize
    \centering
\begin{ThreePartTable}
    \begin{tabular}{l l l}
    \toprule
    \textbf{Tag\textsuperscript{1}} & \textbf{Definition\textsuperscript{2}} & \textbf{Data Source}\\
    \midrule
    \typeA{Origin Relay} & There exists $M_2$ such that $M_2=O_1$. & Self-contained\\
    \typeA{Origin AS-Set} & $O_2$ is in the form of AS-set. & Self-contained\\
    \typeA{Origin Related} & $O_1$ and $O_2$ have a business relationship. & CAIDA~\cite{caida_as_relationship}\\
    \typeA{Private ASN} & The ASN of $O_2$ is reserved for private use. & IANA~\cite{mitchell2013autonomous}\\
    \typeA{Similar Name} & $O_1$ and $O_2$ have similar\textsuperscript{3} organization names. & CAIDA~\cite{caida_as_organization}\\
    \typeB{Direct View} & There exists $M_1$ such that $M_1=V_2$. & Self-contained\\
    \typeB{Country Diff} & $O_1$ and $O_2$ are located in different countries. & CAIDA~\cite{caida_as_organization}\\
    \bottomrule
    \end{tabular}
\begin{tablenotes}
\scriptsize
\item[1] Tags in \sampleA indicate route engineering practices, while tags in \sampleB are informational.
\item[2] The notations are the same as in \S\ref{sec:heuristics}.
\item[3] Two strings are deemed similar if their fuzz partial ratio score is greater than 90.
\end{tablenotes}
\end{ThreePartTable}
\end{table}

\begin{figure}[t]
\captionsetup{belowskip=-2mm}
    \centering
    \includegraphics[width=\linewidth]{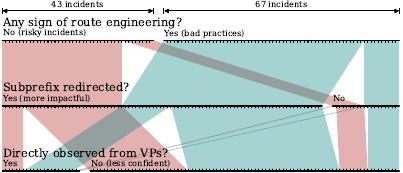}
    \caption{Breakdown of the unique incident set.}
    \label{fig:incidents-breakdown}
\end{figure}

\begin{table}[t]
\definecolor{firebrick}{rgb}{0.698,0.133,0.133}
\newcommand\typeA[1]{{\cellcolor{firebrick!35} #1}}
\newcommand\typeB[1]{{\cellcolor{teal!35} #1}}
\newcommand\typeC[1]{{\cellcolor{white!35} #1}}
    \caption{Overall impact of the incidents.}
    \label{tab:incident-impact}
    \scriptsize
    \centering
\resizebox{\linewidth}{!}{
    \begin{tabular}{c c c c c c}
    \toprule
    \textbf{Type}&\textbf{\#Countries}&\textbf{\#Prefixes}&\textbf{\#Origins}&\textbf{\#Routes}&\textbf{\#VPs} \\
    \midrule
    Risky incidents&\typeA{16}&\typeA{60}&\typeA{36}&\typeA{773}&\typeA{48} \\
    Bad practices&\typeB{24}&\typeB{103}&\typeB{43}&\typeB{3,611}&\typeB{50} \\
    Total&\typeC{31}&\typeC{156}&\typeC{73}&\typeC{4,278}&\typeC{50} \\
    \bottomrule
    \end{tabular}
}
\end{table}

\begin{figure}[t]
\captionsetup{belowskip=-2mm}
    \centering
    \includegraphics[width=\linewidth]{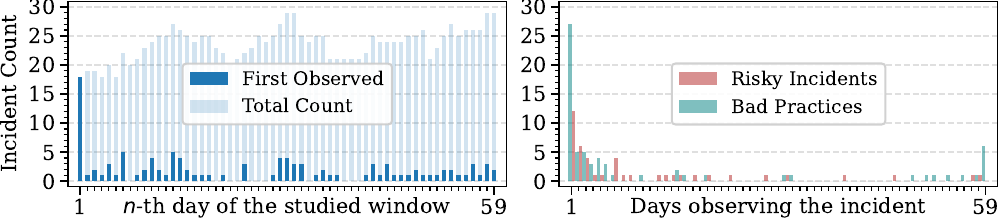}
    \caption{Daily incident count (left) and incident duration (right).}
    \label{fig:daily-incidents}
\end{figure}

\begin{figure}[t]
    \centering
    \includegraphics[width=\linewidth]{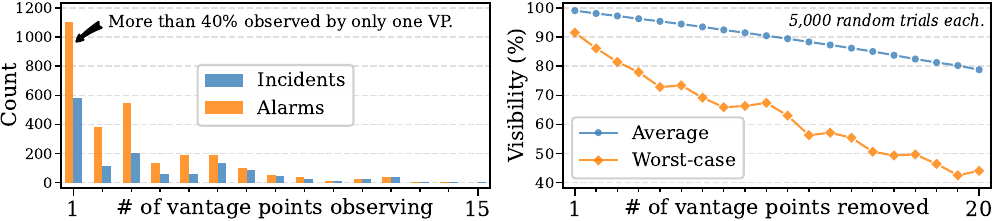}
    \caption{Number of VPs observing each incident (left), and overall incident visibility as VPs are randomly removed (right).}
    \label{fig:vp-distribution}
\end{figure}

\begin{figure}[t]
\captionsetup{belowskip=-2mm}
    \centering
    \includegraphics[width=\linewidth]{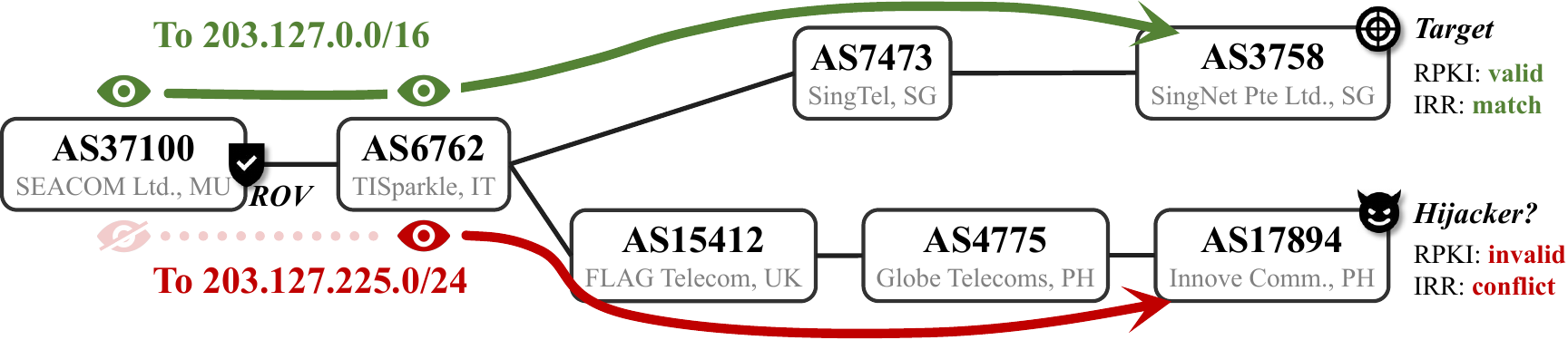}
    \caption{A real-world stealthy hijacking incident.}
    \label{fig:case-study}
\end{figure}

\section{Uncovering Stealthy Hijacking in the Wild}
\label{sec:empirical-study}

In this section, we develop heuristics for stealthy hijacking discovery, and empirically investigate the threat in the wild.

\subsection{Heuristics for Stealthy Hijacking Discovery}
\label{sec:heuristics}

From Definition~\ref{def:stealthy-hijacking-risk}, we derive practical heuristics to discover stealthy hijacking instances based on routing table discrepancies. For clarity, we denote \scalebox{0.9}{$p: V \dotsm (M) \dotsm O$} a route to prefix $p$, where $V$ is the vantage point, $O$ is the origin AS, and $M$ (if present) represents an intermediate AS along the path. Given two routes \scalebox{0.9}{$p_1: V_1 \dotsm (M_1) \dotsm O_1$} and \scalebox{0.9}{$p_2: V_2 \dotsm (M_2) \dotsm O_2$}, we examine the following conditions:
\begin{enumerate}[itemsep=0.1em,align=parleft,left=0pt..14pt]
    \item \underline{Conflict:} $p_2$ equals or is a sub-prefix of $p_1$, and $O_2 \neq O_1$.
    \item \underline{Unauthorized:} $p_2/O_2$ is RPKI-invalid while $p_1/O_1$ is valid.
    \item \underline{Stealthiness:} $V_1$ has no route to $p_2$ originated by $O_2$.
    \item \underline{Risk-critical AS:} There exist $M_1$ and $M_2$, with $M_1=M_2$.
    \item \underline{Risk-critical VP:} There exists $M_1$ such that $M_1=V_2$.
\end{enumerate}

We first define the \emph{loose heuristics}, which require conditions 1-4 to hold simultaneously to yield a ``victim-target-hijacker'' instance $(V_1, O_1, O_2)$, where $M_1$ ($M_2$) is the risk-critical AS. Since condition 4 infers the risk-critical AS from all intermediate ASes, the loose heuristics maximally utilize available route observations. However, the results may not be strictly reliable, as the risk-critical AS's route to the hijacker is inferred from a route segment rather than directly observed. To improve reliability, we further propose the \emph{strict heuristics}, replacing condition 4 with condition 5 to restrict risk-critical ASes to those that are also vantage points. This ensures strict adherence to Definition~\ref{def:stealthy-hijacking-risk} and produces more reliable results.

The two heuristics offer a trade-off between breadth and confidence. In practice, we apply both in our empirical study to maximize discovery while accounting for different confidence levels. For analytical risk assessment, we rely solely on the strict heuristics, as our framework infers complete Internet-wide routes (elaborated later).

\subsection{Real-World Observations and Insights}

We now present our empirical study to track stealthy hijacking in the wild. Specifically, we analyze daily RIB snapshots taken at 12:00 UTC by the RouteViews~\cite{routeviews} collectors \emph{route-views2}, \emph{amsix}, and \emph{wide}, starting from January 1, 2025. These collectors are based in North America, Europe, and Asia, respectively. Each daily archive contains about 50 million BGP routes from over 370 vantage points. For each day's snapshot, we apply the loose heuristics to discover potential stealthy hijacking instances. To ensure prefix-origin legitimacy, we cross-check RIPE NCC's RPKI database~\cite{riperpki}, the RADb IRR database~\cite{radbirr}, and the five RIRs' WHOIS databases~\cite{whois}, retaining only instances where the misbehaving origin is simultaneously RPKI-invalid, IRR-conflicting, and WHOIS-mismatching. To group related instances, we aggregate those affecting the same prefix into a single alarm, then merge alarms with the same misbehaving origin into a single incident.

To profile discovered incidents, we assign each incident with tags corresponding to predefined behaviors that either indicate route engineering or provide additional context. Table~\ref{tab:incident-tags} details the tag definitions and data sources used. Incidents without any tags indicative of route engineering are considered particularly risky, while those linked to route engineering still exhibit stealthy hijacking but are more likely due to misconfigurations, such as improper route aggregation or private IP leasing without updating registries. Notably, incidents tagged \emph{Direct View} follow the strict heuristics and are more reliable.

\noindent\textbf{Findings.}
Over a two-month window starting January 1, 2025, we capture 1,394 potential stealthy hijacking incidents in the wild. Deduplicating them over the timeline yields 110 unique incidents. Figure~\ref{fig:incidents-breakdown} provides a detailed breakdown of them: 43 cases, with no signs of route engineering, are particularly risky, while the rest are attributed to bad operational practices. Among them, 91 involve sub-prefix hijacking, which tends to have more serious impacts, and 22 are directly observed from vantage points, thus of the highest confidence. In total, these incidents involve 4,278 routes observed by 50 vantage points, affecting 156 prefixes and 73 origins across 31 countries, as summarized in Table~\ref{tab:incident-impact}. Over time, we observe 18-29 incidents per day, with 0-5 newly discovered daily (except on Day 1), as shown in Figure~\ref{fig:daily-incidents} (left). In terms of duration, 76 incidents (69.1\%) persist for 7 days or fewer, while 17 (12.7\%) last over 30 days, including 14 deemed bad operational practices, as shown in Figure~\ref{fig:daily-incidents} (right). Moreover, Figure~\ref{fig:vp-distribution} (left) shows that most incidents are seen by three or fewer vantage points, with over 40\% visible to only one. This indicates a strong dependence of stealthy hijacking visibility on vantage point selection. Figure~\ref{fig:vp-distribution} (right) further confirms this: randomly removing just 20 vantage points leads to a 22\% average drop in observable incidents, and up to 55\% in the worst case.

\smallskip
\noindent\textit{\underline{Takeaway:} Stealthy hijacking in the wild is mostly short-lived and targets sub-prefixes, with new cases emerging almost daily and some persisting long-term, likely due to overlooked misconfigurations. Its exposure is sensitive to vantage points.}

\smallskip
\noindent\textbf{Case Study.}
We present an example incident in Figure~\ref{fig:case-study}. This incident persists throughout our study and shows no signs of route engineering. Both vantage points, AS37100 (\emph{SEACOM}) and AS6762 (\emph{TISparkle}), observe the prefix 203.127.0.0/16 announced by its legitimate origin, AS3758 (\emph{SingNet}). Meanwhile, the sub-prefix 203.127.225.0/24 is announced by an unauthorized origin, AS17894 (\emph{Innove Comm.}). The two origins are in different countries and have no established relationship. Since AS37100 has deployed ROV with a 100\% filtering rate~\cite{apnic}, it discards the bogus /24 route. Yet, traffic from AS37100 or its customers to the /24 prefix still experiences hijacking when transited through legacy AS6762, which accepts the bogus route (the red path). Examination of AS37100's looking glass further confirms the incident (detailed in Appendix~\ref{sec:appendix:case-study}). We reported it to AS4775 (\emph{Globe Tel.}) on February 10, 2025, and received promise to investigate.

\noindent\textbf{Resources.}
We curate a high-confidence dataset of 318 real-world stealthy hijacking incidents, covering 2,178 unique routes, by selecting incidents tagged \emph{Direct View} while excluding those with \emph{Similar Name}\footnote{Incidents tagged \emph{Similar Name} are excluded to remove cases that are likely caused by private interconnection between affiliated ASes~\cite{giotsas2014inferring}.}. This dataset serves as ground truth for broader research and is specifically used to validate our framework in \S\ref{sec:performance-evaluation}. Beyond this study, we continue to run stealthy hijacking discovery as a service, implement on-demand data-plane validation based on RIPE Atlas~\cite{ripeatlas}, and provide a feature-rich frontend to publish daily incident reports. Readers are encouraged to explore these resources at \url{https://yhchen.cn/stealthy-bgp-hijacking}.

We conclude this section by reiterating a key observation, \ie stealthy hijacking exposure is sensitive to vantage point selection. This highlights the need for a comprehensive view of global routing to enable deterministic and exhaustive risk assessment, motivating our analytical approach presented next.
\section{The \ours Framework}
\label{sec:analytical-approach}

\begin{figure*}[t]
\captionsetup{belowskip=-2mm}
    \centering
    \includegraphics[width=.9\linewidth]{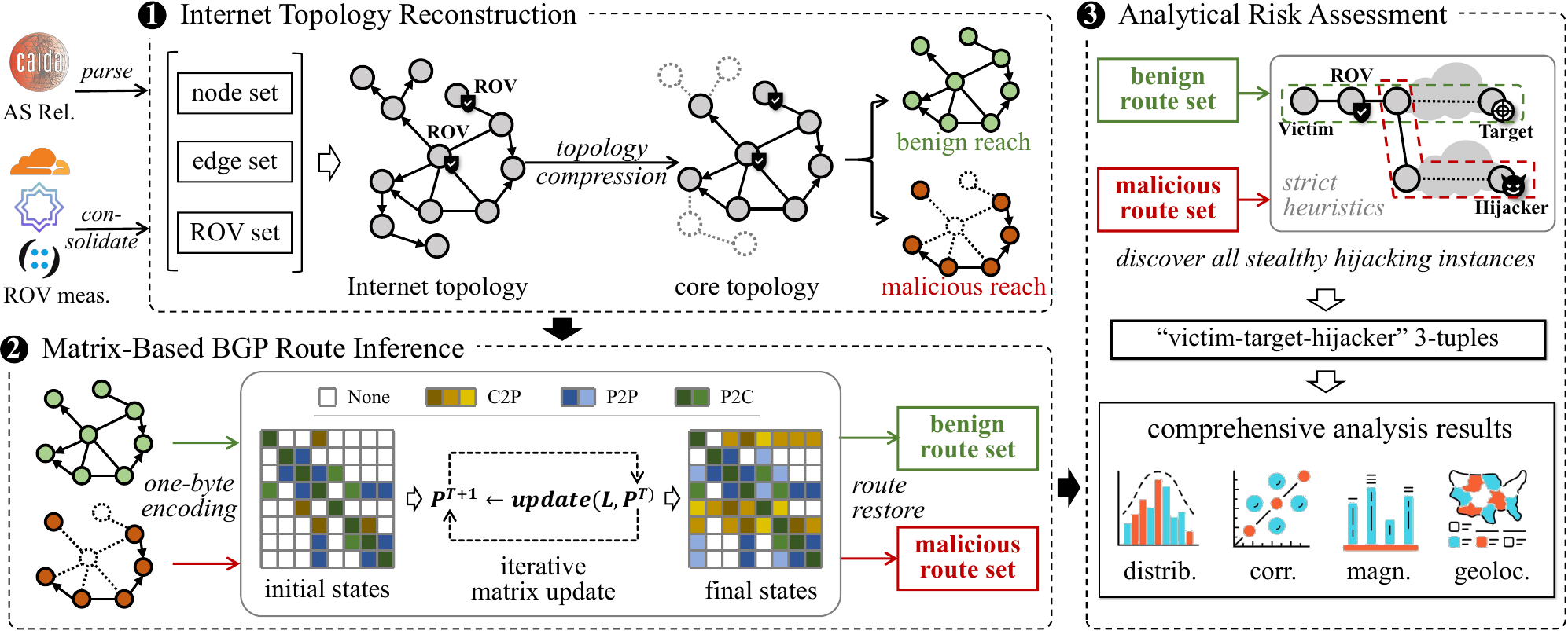}
    \caption{Workflow of the \ours framework.}
    \label{fig:workflow}
\end{figure*}

{\centering\noindent\oursmotto}

\subsection{Overview}
We present \ours, a BGP route inference framework dedicated to analytical assessment of stealthy hijacking risk. As shown in Figure~\ref{fig:workflow}, \ours takes AS relationships and ROV measurements from multiple sources as input. It then \first reconstructs the Internet topology, \second performs matrix-based BGP route inference, and \third applies the strict heuristics to assess stealthy hijacking risk across the inferred routes.

The core rationale is to determine whether any AS along legitimate routes can forward traffic to potential hijackers. Achieving comprehensive risk assessment thus requires knowledge of both \emph{all} routes to benign ASes and \emph{all} routes to potential hijackers, posing a primary challenge given the Internet scale. To address this, \ours compresses the Internet topology during reconstruction and extracts the benign reach and malicious reach, \ie the respective sub-topologies traversable by benign and malicious announcements (\hyperref[fig:workflow]{see~\scalebox{0.8}{\circled{\textbf{1}}}}). By converting them into matrices, where each cell encodes compact routing information, \ours iteratively updates these matrices through highly optimized matrix operations to infer Internet-wide routes (\hyperref[fig:workflow]{see~\scalebox{0.8}{\circled{\textbf{2}}}}). It then identifies risk-critical ASes by checking each benign route with the malicious route set (\hyperref[fig:workflow]{see~\scalebox{0.8}{\circled{\textbf{3}}}}). We elaborate on each step below.

\subsection{Internet Topology Reconstruction}
\label{sec:topology-reconstruction}

To reconstruct the Internet topology $G=(V,E,V_{rov})$, we use the CAIDA AS relationship dataset~\cite{caida_as_relationship} to obtain the set of ASes ($V$) and AS-to-AS links ($E$). Each link $(a_i, a_j, r) \in E$, where $a_i, a_j \in V$, is associated with a relationship type $r \in \{\text{C2P}, \text{P2P}, \text{P2C}\}$. To determine the set of ROV-enabled ASes ($V_{rov}$), we consolidate measurements from three sources: APNIC~\cite{apnic}, RoVista~\cite{li2023rovista}, and Cloudflare~\cite{cloudflare}. APNIC and RoVista report ROV-filtering rates per covered AS, and we consider an AS to be ROV-enabled if its filtering rate is at least 80\%, a confidence threshold adopted in prior work~\cite{chen2022rov} and also proved practical in our own experiments. Cloudflare, meanwhile, directly provides a list of ROV-enabled ASes. To maximize coverage, we include in $V_{rov}$ any AS identified as ROV-enabled by at least one of these sources.

We next compress the resulting topology. We notice that, under certain conditions, the routing table of a single-homed AS can be directly derived from its upstream transit AS, which exchanges all routes with it. As such, there is no need to include these single-homed ASes in the computationally intensive route inference process; instead, we remove them to form a \emph{core topology} for route inference, and later recover their routing tables from their respective upstreams. This preserves the integrity of inferred routes while reducing the number of vertices by 36.3\%, greatly improving inference efficiency. We defer details of the compression method in Appendix~\ref{sec:appendix:topology-compression}.

From the core topology, we further derive the benign reach and the malicious reach. The benign reach, where benign announcements can propagate freely, is identical to the core topology, while the malicious reach, denoting the restricted area where malicious announcements can propagate as ROV-enabled ASes block them, is obtained by removing all ROV-enabled ASes from the core topology. Route inference is later performed on these reaches to obtain legitimate AS-level routes and potential AS-to-hijacker routes, respectively.

\subsection{Matrix-Based BGP Route Inference}
\label{sec:matrix-based-bgp-simulation}

We leverage highly optimized matrix operations for efficient BGP route inference, addressing the limitations of existing methods~\cite{brandt2021optimized,furuness2023bgpy,li2023realizing}, which incur prohibitive overhead in generating complete BGP routes at the Internet scale.

\noindent\textbf{Inference Criteria.}
We follow the same criteria commonly used in prior works~\cite{morillo2021rov++,gilad2016we,mcdaniel2020flexsealing}, which are derived from the Gao-Rexford model~\cite{gao2001stable} and capture typical BGP control-plane behaviors. Specifically, the best-route selection follows a three-step process: \first prefer routes with the highest local preference, \second select routes with the shortest AS\_path, and \third break ties randomly. By default, local preference is set to reflect business incentives, \ie a route received from a customer is preferred over one from a peer, which is preferred over one from a provider. Additionally, the valley-free constraint~\cite{gao2001stable} is enforced in route propagation, \ie routes learned from a peer or provider are only forwarded to customers.

\noindent\textbf{One-Byte Route Priority Encoding.} Best-route selection is computationally intensive. To boost the process, we propose a one-byte route priority encoding that enables fast comparison and efficient update. Specifically, the priority of a route during selection is determined by two key properties, \ie local preference and path length. We use one byte to represent both:
\begin{equation*}
\newcommand\drawbits[1]{
    \tikz[x=1.5ex,y=1.5ex,every path/.style={draw=black,semithick}]{
        \foreach \y [count=\i] in {#1} {
            \node[] at (\i+0.5,0.5) {\small\y};
            \draw (\i,0) rectangle (\i+1,1);
        }
    }
}
\newcommand\scalemath[2]{\scalebox{#1}{\mbox{\ensuremath{\displaystyle #2}}}}
\scalemath{1}{
\overbrace{
    \underbrace{\text{\LARGE\drawbits{$b_7$,$b_6$}}}_{\text{\small LP}} \quad \underbrace{\text{\LARGE\drawbits{$b_5$,$b_4$,$b_3$,$b_2$,$b_1$,$b_0$}}}_{\text{\small \textasciitilde PL}}
}^{\text{\normalsize 8 bits in a byte}}
}
\end{equation*}
\noindent where the LP field (the two most significant bits) encodes the exact value of local preference, and the \textasciitilde PL field (the six least significant bits) encodes the \emph{bitwise complement} value of path length. The bit values of LP and PL are defined as follows:
\begin{align}
\label{eq:lp-definition}
    \text{LP} & = \left\{
        \begin{array}{rl}
            (11)_2, & \text{for routes received from customers} \\
            (10)_2, & \text{for routes received from peers} \\
            (01)_2, & \text{for routes received from providers} \\
            (00)_2, & \text{for unreachable origins}
        \end{array}
    \right. \\
    \label{eq:pl-definition}
    \text{PL} &= \left\{
    \begin{array}{ll}
        l \text{ in 6-bit form}, & \text{for path length of } l \\
        (000000)_2, & \text{for unreachable origins}
    \end{array}
    \right.
\end{align}
\noindent For example, a route received from a customer with a path length of 5 has $\text{LP} = (11)_2$ and $\text{PL} = (000101)_2$, resulting in a priority byte of $(11\ 111010)_2$. If the origin is unreachable or the route is yet to know, we use $(00\ 111111)_2$ as a placeholder for future update. Note that the 6-bit \textasciitilde PL field can represent a path length up to 63, so any longer route will be dropped. In our practice, no route exceeds this limit.

Compared with traditional data structures, our one-byte encoding maximizes memory utilization and enables fast priority comparison, \ie a greater byte value indicates a more preferable route. As a result, best-route selection becomes a simple task of finding the maximum value in a byte array. Moreover, its update during route propagation involves only basic arithmetic, \eg by subtracting one per hop to update the path length field (\textasciitilde PL). Such properties, as detailed later, help break down the complex iterations of BGP route inference into a series of highly optimized matrix operations.

\noindent\textbf{Route Priority Update.}
We follow an iterative paradigm for BGP route inference. It starts with an initial state where no AS knows a route to any other. As each AS announces its presence, route announcements propagate hop by hop across the Internet topology. Each round of one-hop propagation, along with corresponding routing table updates, is \emph{an iteration}.

We now describe how to update the one-byte encoding during an iteration. Consider a route to origin AS $a_j$ in AS $a_k$'s routing table: after $T$ iterations, we denote its priority byte by $p_{kj}^{T}$. In the next iteration, this route is forwarded one hop further. Suppose AS $a_i$ is to receive this route, then the priority byte of the \emph{received} route, denoted by $\hat{p}_{ikj}^{T\textsl{+1}}$, is determined by the $a_i$-to-$a_k$ relationship type and the value of $p_{kj}^T$. If the $a_i$-to-$a_k$ relationship exists and the propagation satisfies the valley-free constraint, $\hat{p}_{ikj}^{T\textsl{+1}}$ is set accordingly: its LP field inherits the local preference value for the $a_i$-to-$a_k$ relationship type (\ie the LP value of $p_{ik}^\textsl{1}$), and its \textasciitilde PL field equals $p_{kj}^T$'s \textasciitilde PL value minus one due to the increased path length. However, if the $a_i$-to-$a_k$ relationship does not exist or the propagation violates the valley-free constraint, $a_i$ is then not eligible to receive the route. In this case, we assign the placeholder value to $\hat{p}_{ikj}^{T\textsl{+1}}$.

Intuitively, this process involves two steps: \first compute the LP field of $\hat{p}_{ikj}^{T\textsl{+1}}$, and \second assign the \textasciitilde PL field accordingly. For the first step, we introduce a custom operator $\odot$:
\begin{equation}\label{eq:lp-update}
    \text{LP}[\hat{p}_{ikj}^{T\textsl{+1}}] = \text{LP}[p_{ik}^\textsl{1}] \odot \text{LP}[p_{kj}^T],
\end{equation}
\noindent where $\text{LP}[*]$ denotes the LP field of the corresponding priority byte. The $\odot$ operator maps all scenarios described above to corresponding outcomes; its truth table is shown in Table~\ref{tab:lp-update-truth-table}.

\begin{table}[ht]
\setlength\tabcolsep{3pt}
\renewcommand\arraystretch{1}
\newcommand\nonelk[1]{{\cellcolor{gray!20} #1}} 
\newcommand\nonevf[1]{{\cellcolor{BrickRed!20} #1}}
\newcommand\otherblk[1]{{\cellcolor{white} #1}}
\newcommand\samplelk{{\protect\tikz[x=3ex,y=1.5ex] \protect\fill [gray!20] (0,0) rectangle (1,1);}\ }
\newcommand\samplevf{{\protect\tikz[x=3ex,y=1.5ex] \protect\fill [BrickRed!20] (0,0) rectangle (1,1);}\ }
\caption{\label{tab:lp-update-truth-table} Truth table of $\text{LP}[p_{ik}^\textsl{1}] \odot \text{LP}[p_{kj}^T]$. The 00 outputs in \samplelk indicate that the input route does not exist or is not yet known, while those in \samplevf result from valley-free violations.}
\centering
\footnotesize
\begin{tabular}{|l|c c c c|}
    \hline
    \diagbox{$\text{LP}[p_{ik}^\textsl{1}]$}{$\text{LP}[p_{kj}^T]$}
       & \makecell{11\\(P2C)} & \makecell{10\\(P2P)} & \makecell{01\\(C2P)} & \makecell{00\\(None)} \\\hline
    11 (P2C) & \otherblk{11} & \nonevf{00} & \nonevf{00} & \nonelk{00} \\
    10 (P2P)& \otherblk{10} & \nonevf{00} & \nonevf{00} & \nonelk{00} \\
    01 (C2P)& \otherblk{01} & \otherblk{01} & \otherblk{01} & \nonelk{00} \\
    00 (None)& \nonelk{00} & \nonelk{00} & \nonelk{00} & \nonelk{00} \\\hline
\end{tabular}
\end{table}

The $\odot$ operator can be decomposed into basic arithmetic operations. Let $(ab)_2$, $(cd)_2$, and $(ef)_2$ denote the two bits in $\text{LP}[p_{ik}^\textsl{1}]$, $\text{LP}[p_{kj}^T]$, and $\text{LP}[p_{ik}^\textsl{1}] \odot \text{LP}[p_{kj}^T]$, respectively. Then, Table~\ref{tab:lp-update-truth-table} can be expressed by minimal boolean expressions:
\begin{align}
    \label{eq:boolean-e}
    e & = acd = a \cdot cd + 0 \cdot 0, \\
    \label{eq:boolean-f}
    f & = \bar{a}bd+\bar{a}bc+bcd = b \cdot cd + \bar{a}b \cdot (c+d).
\end{align}
\noindent We can thus implement $\odot$ efficiently
using shift and bitwise logic operations, which are vectorizable over matrices.

With $\hat{p}_{ikj}^{T\textsl{+1}}$'s LP field computed, its \textasciitilde PL field is determined:
\begin{equation}\label{eq:pl-update}
    \text{\textasciitilde PL}[\hat{p}_{ikj}^{T\textsl{+1}}] =
    \left\{
    \begin{array}{ll}
    \text{\textasciitilde PL}[p_{kj}^T] - 1, & \text{if } \text{LP}[\hat{p}_{ikj}^{T\textsl{+1}}] \neq (00)_2  \\
    (111111)_2, &  \text{otherwise}                     
    \end{array}
    \right.
\end{equation}
\noindent where $\text{LP}[*]$ and $\text{\textasciitilde PL}[*]$ denote the respective fields. Note that this two-branch computation can be further reformulated as a branchless expression without affecting correctness:
\begin{equation}\label{eq:pl-update-non-branch}
    \text{\textasciitilde PL}[\hat{p}_{ikj}^{T\textsl{+1}}] = \text{\textasciitilde PL}[p_{kj}^T] - \text{PL}[p_{ik}^\textsl{1}],
\end{equation}
\noindent where $\text{PL}[*]$ denotes the \emph{exact value} of path length. We provide the proof of equivalence between Equations~\eqref{eq:pl-update} and \eqref{eq:pl-update-non-branch} in Appendix~\ref{sec:appendix:non-branch-validity}. Since Equation~\eqref{eq:pl-update-non-branch} is independent of the LP field computation and is more amenable to vectorization, we implement Equation~\eqref{eq:pl-update-non-branch} in practice.

As $\hat{p}_{ikj}^{T\textsl{+1}}$ indicates an updated route that $a_i$ receives from $a_k$, the best-route selection at $a_i$ involves comparing $\hat{p}_{ikj}^{T\textsl{+1}}$ with $p_{ij}^T$ and prefers the one with the higher value. The priority byte of the final selected route at the end of iteration $T\textsl{+1}$ is thus determined by considering all $\hat{p}_{ikj}^{T\textsl{+1}}$ values across $k$:
\begin{equation}\label{eq:byte-selection}
    p_{ij}^{T\textsl{+1}} = max \left\{p_{ij}^T, \hat{p}_{i\textsl{0}j}^{T\textsl{+1}}, \dots, \hat{p}_{i\textsl{(n-1)}j}^{T\textsl{+1}} \right \},
\end{equation}
\noindent where $p_{ij}^T$ is known from iteration $T$ and $\hat{p}_{ikj}^{T\textsl{+1}}, k=0,\dots,n-1$ are computed via Equations~\eqref{eq:lp-update}-\eqref{eq:pl-update-non-branch}. We further prove that there always exists some $k$ such that $\hat{p}_{ikj}^{T\textsl{+1}} \geq p_{ij}^T$ (see Appendix~\ref{sec:appendix:best-route-selection-simplification}). Thus, the best-route selection process can be simplified as:
\begin{equation}\label{eq:byte-selection-simplified}
    p_{ij}^{T\textsl{+1}} = max\{ \hat{p}_{ikj}^{T\textsl{+1}} \}_{k=0}^{n\textsl{-1}}.
\end{equation}
\indent Based on Equation~\eqref{eq:byte-selection-simplified}, we can update $p_{ij}^{T}$ to $p_{ij}^{T\textsl{+1}}$ for \emph{each} pair of $i$ and $j$. To accomplish this efficiently, we design a matrix representation scheme to organize all route priority bytes, and extend byte-wise operations to equivalent matrix-wise ones that are highly optimized for batch processing.

\noindent\textbf{Matrix Representation.}
Extending the variable $p_{ij}^{T}$ across all $(i, j)$ pairs naturally forms an $n \times n$ matrix, where the cell at $i$-th row and $j$-th column contains the corresponding byte $p_{ij}^T$. We thus define a state matrix $P^T$ to collectively maintain all route priority bytes and update it recursively in place:
\begin{equation}\label{eq:matrix-recurrence}
    P^{T\textsl{+1}} = \textproc{Update}(P^\textsl{1}, P^T),
\end{equation}
\noindent where $T$ starts at 0, and the generating function \textproc{Update} performs the per-iteration computation on all route priority bytes as defined by Equations~\eqref{eq:lp-update} to \eqref{eq:byte-selection-simplified}. The initial state $P^\textsl{0}$ is constructed by setting all diagonal cells to $(11\ 111111)_2$ and all off-diagonal cells to $(00\ 111111)_2$, since each AS has no external reachability knowledge other than their own existence at the start of route inference. Note that the diagonal cells indicate conceptually self-pointing routes, by which an AS can reach itself without traversing any other ASes. As such, these routes have the highest local preference and a path length of zero, and are always preferred over looped routes.

Besides self-pointing routes, one-hop routes from each AS to its neighbors form another prior knowledge. $P^\textsl{1}$ captures these routes and is referenced during each iteration. Therefore, we treat $P^\textsl{1}$ as a constant and pre-compute it using a separate matrix $L$ (short for ``Link''). Equation~\eqref{eq:matrix-recurrence} is then altered:
\begin{equation}\label{eq:matrix-recurrence-altered}
    P^{T\textsl{+1}} = \textproc{Update}(L, P^T),
\end{equation}
where $L$ is initialized based on the given $G=(V, E, V_{rov})$:
\begin{equation}\label{eq:L-definition}
    L: l_{ij}=
    \left\{
    \begin{array}{ll}
    (11\ 111111)_2, & \text{for } i = j \\
    (11\ 111110)_2, & \text{for } (a_i,a_j,\text{P2C}) \in E \\
    (10\ 111110)_2, & \text{for } (a_i,a_j,\text{P2P}) \in E \\
    (01\ 111110)_2, & \text{for } (a_i,a_j,\text{C2P}) \in E \\
    (00\ 111111)_2, &  \text{otherwise}                     
    \end{array}
    \right.
\end{equation}
\indent The initialization of $P^\textsl{0}$ and $L$ follows the one-byte encoding. It can also be verified that $L=P^\textsl{1}=\textproc{Update}(L, P^\textsl{0})$ by evaluating the first iteration. In this matrix form, the update of priority bytes can be naturally extended to matrix-wise operations. Additional technical details are provided in Appendix~\ref{sec:appendix:additional-reports}.

\noindent\textbf{Route Set Generation.}
Through matrix operations, we iteratively update $P^T$ until no changes occur. During this process, we also record the next-hop AS for each selected route. Specifically, we maintain an $n\times n$ matrix $N^T$, where the cell at the $i$-th row and the $j$-th column stores the index of the next-hop AS for $a_i$ to reach $a_j$ after $T$ iterations. $N^T$ is updated alongside $P^T$. As iteration $T$ yields the maximum priority byte for $p_{ij}^{T\textsl{+1}}$ according to Equation~\eqref{eq:byte-selection-simplified}, we simultaneously record the index $k$ of the selected byte in the corresponding cell of $N^{T\textsl{+1}}$. Once inference completes, the full AS paths can be restored by recursively tracing the next-hop AS using $N^T$ (detailed in Appendix~\ref{sec:appendix:restoring-route-from-next-hop}). In this way, we generate the complete set of AS-level routes across the Internet topology: On the benign reach, we obtain the benign route set, and on the malicious reach, we obtain the malicious route set.

\subsection{Analytical Risk Assessment}

Given the benign and malicious route sets, we systematically discover all potential stealthy hijacking instances. According to our strict heuristics, stealthy hijacking becomes possible only if the victim does not receive any malicious route to the hijacker (otherwise it is direct hijacking), while at least one intermediate AS on the benign victim-to-target route accepts a malicious route to the hijacker. Therefore, we iterate through all benign routes, treating the vantage point AS as the victim and the origin AS as the target. We then check all ASes to identify potential hijackers that satisfy these conditions:
\begin{itemize}[itemsep=.3em,align=parleft,left=2pt..13pt]
    \item The hijacker is neither the victim nor the target.
    \item No victim-to-hijacker route exists in the malicious set.
    \item There is an AS on the victim-to-target route for which the route to the hijacker exists in the malicious route set.
\end{itemize}
\noindent By exhaustively examining all benign routes, we obtain a complete set of stealthy hijacking instances that \emph{could} occur in the current Internet. Based on these instances, we can assess the risk in a statistical manner, which is presented next.

\section{Assessing the Stealthy Hijacking Risk}
\label{sec:risk-assessment}

In this section, we assess the stealthy hijacking risk posed by the current ROV deployment through \ours.

\subsection{Framework Setup}

\begin{figure*}[t]
\setlength{\belowcaptionskip}{-4mm}
    \centering
    \includegraphics[width=.88\linewidth]{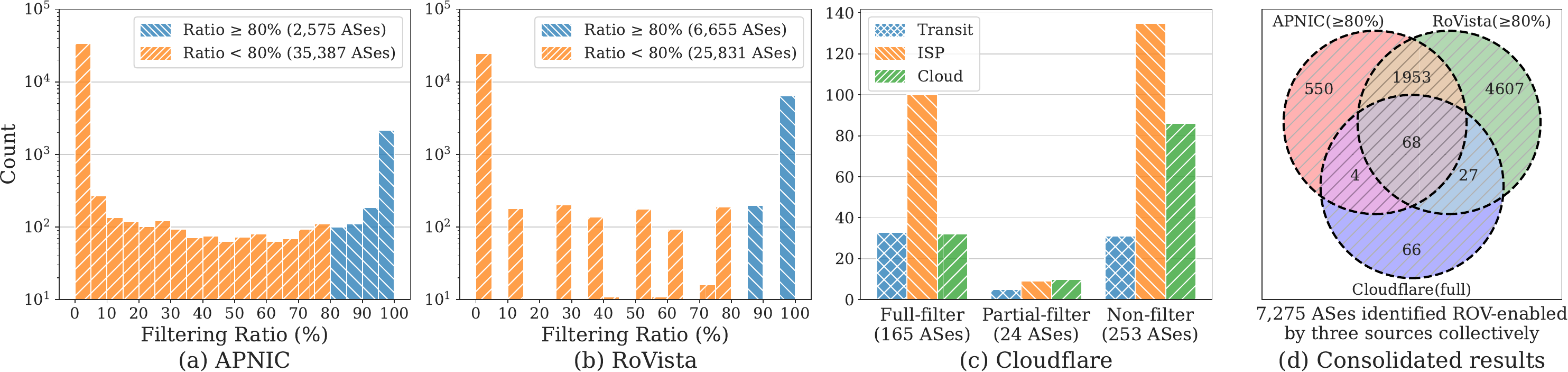}
    \caption{The three ROV measurements and the consolidated results.}
    \label{fig:rov-measurements}
\end{figure*}

We implement \ours in Python 3.10 and use Numba to compile matrix operations to low-level C code with GPU acceleration. It performs on a Linux platform with an Intel Xeon E5-2650 CPU and an NVIDIA GeForce RTX 2080Ti GPU.
For Internet topology reconstruction, we use the CAIDA AS relationship dataset~\cite{caida_as_relationship} released on March 1, 2025, which contains 77,600 ASes and 709,737 AS relationships. The three ROV measurements are collected on the same date. APNIC~\cite{apnic} reports 43,042 ASes, of which 2,575 exhibit an ROV filtering ratio of at least 80\%. RoVista~\cite{li2023rovista} covers 32,486 ASes, with 6,655 meeting the same threshold. In addition, Cloudlare lists 165 ASes with full ROV filtering. Figure~\ref{fig:rov-measurements} summarizes their statistics. Collectively, we identify 7,275 ASes recognized as ROV-enabled by at least one source. This set forms our final input of ROV-enabled ASes to \ours.

Using the AS relationship data and ROV measurements, we reconstruct the Internet topology. After compression, the resulting core topology contains 49,403 ASes (a 36.3\% reduction) and 681,540 AS relationships (a 3.97\% reduction). From this core topology, we derive both the benign reach and the malicious reach, and perform matrix-based BGP route inference on each. We run the inference for up to 20 iterations to ensure convergence. This generates the benign route set with 5,963,253,322 routes and the malicious route set with 2,399,622,350 routes. Based on these inferred routes, we discover the complete set of potential stealthy hijacking instances and conduct further assessment.

\begin{figure}[t]
\setlength{\belowcaptionskip}{-4mm}
    \centering
    \includegraphics[width=.93\linewidth]{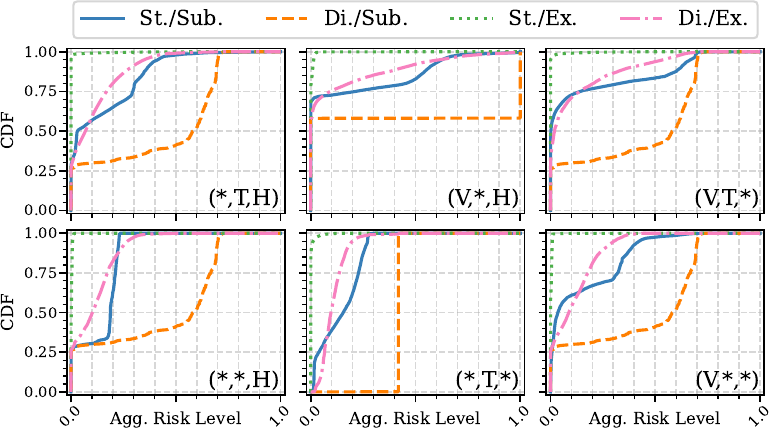}
    \caption{CDF of aggregated risk levels.}
    \label{fig:aggregated-risk-levels}
\end{figure}

\begin{table}[t]
\setlength{\belowcaptionskip}{0.3\baselineskip}
\caption{Key statistics of aggregated risk levels.}
\newcommand\placeholder{$0.830_{\textcolor{blue}{\scriptstyle\blacktriangledown \mathbf{0.123}}}$}
\label{tab:risk-dissection}
\footnotesize
\centering
\begin{ThreePartTable}
\begin{tabu}{X[1,c]  X[0.5,c] | X[1,c] X[1,c] X[1,c] X[1,c]}
    \toprule
    \multicolumn{2}{c|}{\multirow{2}{*}{\textbf{Statistics}\textsuperscript{1}}}&\multicolumn{4}{c}{\textbf{Hijacking Type}\textsuperscript{2}}\\
    \cline{3-6}
    & & \textbf{St./Sub.} & \textbf{Di./Sub.} & \textbf{St./Ex.} & \textbf{Di./Ex.}\\
    \midrule
\multirow{5}{*}{\textbf{$\mathcal{P}$(*,T,H)}}&$\mathit{min}$&$0.000_{\textcolor{blue}{\scriptstyle\blacktriangledown \mathbf{0.000}}}$&$0.000_{\textcolor{blue}{\scriptstyle\blacktriangledown \mathbf{0.000}}}$&$0.000_{\textcolor{blue}{\scriptstyle\blacktriangledown \mathbf{0.000}}}$&$0.000_{\textcolor{blue}{\scriptstyle\blacktriangledown \mathbf{0.000}}}$\\
&$\mathit{25th}$&$0.000_{\textcolor{BrickRed}{\scriptstyle\blacktriangle \mathbf{0.000}}}$&$0.001_{\textcolor{blue}{\scriptstyle\blacktriangledown \mathbf{0.992}}}$&$0.000_{\textcolor{blue}{\scriptstyle\blacktriangledown \mathbf{0.000}}}$&$0.000_{\textcolor{blue}{\scriptstyle\blacktriangledown \mathbf{0.183}}}$\\
&$\mathit{50th}$&$0.033_{\textcolor{BrickRed}{\scriptstyle\blacktriangle \mathbf{0.033}}}$&$0.582_{\textcolor{blue}{\scriptstyle\blacktriangledown \mathbf{0.412}}}$&$0.000_{\textcolor{BrickRed}{\scriptstyle\blacktriangle \mathbf{0.000}}}$&$0.069_{\textcolor{blue}{\scriptstyle\blacktriangledown \mathbf{0.246}}}$\\
&$\mathit{75th}$&$0.296_{\textcolor{BrickRed}{\scriptstyle\blacktriangle \mathbf{0.296}}}$&$0.668_{\textcolor{blue}{\scriptstyle\blacktriangledown \mathbf{0.326}}}$&$0.000_{\textcolor{BrickRed}{\scriptstyle\blacktriangle \mathbf{0.000}}}$&$0.170_{\textcolor{blue}{\scriptstyle\blacktriangledown \mathbf{0.329}}}$\\
&$\mathit{max}$&$0.995_{\textcolor{BrickRed}{\scriptstyle\blacktriangle \mathbf{0.995}}}$&$0.737_{\textcolor{blue}{\scriptstyle\blacktriangledown \mathbf{0.259}}}$&$0.993_{\textcolor{BrickRed}{\scriptstyle\blacktriangle \mathbf{0.993}}}$&$0.737_{\textcolor{blue}{\scriptstyle\blacktriangledown \mathbf{0.259}}}$\\
\hline
\multirow{5}{*}{\textbf{$\mathcal{P}$(V,*,H)}}&$\mathit{min}$&$0.000_{\textcolor{blue}{\scriptstyle\blacktriangledown \mathbf{0.000}}}$&$0.000_{\textcolor{blue}{\scriptstyle\blacktriangledown \mathbf{0.000}}}$&$0.000_{\textcolor{blue}{\scriptstyle\blacktriangledown \mathbf{0.000}}}$&$0.000_{\textcolor{blue}{\scriptstyle\blacktriangledown \mathbf{0.000}}}$\\
&$\mathit{25th}$&$0.000_{\textcolor{blue}{\scriptstyle\blacktriangledown \mathbf{0.000}}}$&$0.000_{\textcolor{blue}{\scriptstyle\blacktriangledown \mathbf{1.000}}}$&$0.000_{\textcolor{blue}{\scriptstyle\blacktriangledown \mathbf{0.000}}}$&$0.000_{\textcolor{blue}{\scriptstyle\blacktriangledown \mathbf{0.125}}}$\\
&$\mathit{50th}$&$0.000_{\textcolor{blue}{\scriptstyle\blacktriangledown \mathbf{0.000}}}$&$0.000_{\textcolor{blue}{\scriptstyle\blacktriangledown \mathbf{1.000}}}$&$0.000_{\textcolor{blue}{\scriptstyle\blacktriangledown \mathbf{0.000}}}$&$0.000_{\textcolor{blue}{\scriptstyle\blacktriangledown \mathbf{0.296}}}$\\
&$\mathit{75th}$&$0.218_{\textcolor{BrickRed}{\scriptstyle\blacktriangle \mathbf{0.218}}}$&$1.000_{\textcolor{blue}{\scriptstyle\blacktriangledown \mathbf{0.000}}}$&$0.000_{\textcolor{blue}{\scriptstyle\blacktriangledown \mathbf{0.000}}}$&$0.091_{\textcolor{blue}{\scriptstyle\blacktriangledown \mathbf{0.464}}}$\\
&$\mathit{max}$&$0.994_{\textcolor{BrickRed}{\scriptstyle\blacktriangle \mathbf{0.994}}}$&$1.000_{\textcolor{blue}{\scriptstyle\blacktriangledown \mathbf{0.000}}}$&$0.994_{\textcolor{BrickRed}{\scriptstyle\blacktriangle \mathbf{0.994}}}$&$1.000_{\textcolor{blue}{\scriptstyle\blacktriangledown \mathbf{0.000}}}$\\
\hline
\multirow{5}{*}{\textbf{$\mathcal{P}$(V,T,*)}}&$\mathit{min}$&$0.000_{\textcolor{blue}{\scriptstyle\blacktriangledown \mathbf{0.000}}}$&$0.000_{\textcolor{blue}{\scriptstyle\blacktriangledown \mathbf{0.000}}}$&$0.000_{\textcolor{blue}{\scriptstyle\blacktriangledown \mathbf{0.000}}}$&$0.000_{\textcolor{blue}{\scriptstyle\blacktriangledown \mathbf{0.000}}}$\\
&$\mathit{25th}$&$0.000_{\textcolor{BrickRed}{\scriptstyle\blacktriangle \mathbf{0.000}}}$&$0.001_{\textcolor{blue}{\scriptstyle\blacktriangledown \mathbf{0.992}}}$&$0.000_{\textcolor{blue}{\scriptstyle\blacktriangledown \mathbf{0.000}}}$&$0.000_{\textcolor{blue}{\scriptstyle\blacktriangledown \mathbf{0.068}}}$\\
&$\mathit{50th}$&$0.002_{\textcolor{BrickRed}{\scriptstyle\blacktriangle \mathbf{0.002}}}$&$0.582_{\textcolor{blue}{\scriptstyle\blacktriangledown \mathbf{0.412}}}$&$0.000_{\textcolor{blue}{\scriptstyle\blacktriangledown \mathbf{0.000}}}$&$0.020_{\textcolor{blue}{\scriptstyle\blacktriangledown \mathbf{0.253}}}$\\
&$\mathit{75th}$&$0.141_{\textcolor{BrickRed}{\scriptstyle\blacktriangle \mathbf{0.141}}}$&$0.668_{\textcolor{blue}{\scriptstyle\blacktriangledown \mathbf{0.326}}}$&$0.000_{\textcolor{blue}{\scriptstyle\blacktriangledown \mathbf{0.000}}}$&$0.139_{\textcolor{blue}{\scriptstyle\blacktriangledown \mathbf{0.471}}}$\\
&$\mathit{max}$&$0.740_{\textcolor{BrickRed}{\scriptstyle\blacktriangle \mathbf{0.740}}}$&$0.737_{\textcolor{blue}{\scriptstyle\blacktriangledown \mathbf{0.259}}}$&$0.511_{\textcolor{BrickRed}{\scriptstyle\blacktriangle \mathbf{0.511}}}$&$0.737_{\textcolor{blue}{\scriptstyle\blacktriangledown \mathbf{0.259}}}$\\
\hline
\multirow{5}{*}{\textbf{$\mathcal{P}$(*,*,H)}}&$\mathit{min}$&$0.000_{\textcolor{blue}{\scriptstyle\blacktriangledown \mathbf{0.000}}}$&$0.000_{\textcolor{blue}{\scriptstyle\blacktriangledown \mathbf{0.000}}}$&$0.000_{\textcolor{blue}{\scriptstyle\blacktriangledown \mathbf{0.000}}}$&$0.000_{\textcolor{blue}{\scriptstyle\blacktriangledown \mathbf{0.000}}}$\\
&$\mathit{25th}$&$0.000_{\textcolor{BrickRed}{\scriptstyle\blacktriangle \mathbf{0.000}}}$&$0.001_{\textcolor{blue}{\scriptstyle\blacktriangledown \mathbf{0.992}}}$&$0.000_{\textcolor{blue}{\scriptstyle\blacktriangledown \mathbf{0.000}}}$&$0.000_{\textcolor{blue}{\scriptstyle\blacktriangledown \mathbf{0.260}}}$\\
&$\mathit{50th}$&$0.188_{\textcolor{BrickRed}{\scriptstyle\blacktriangle \mathbf{0.188}}}$&$0.582_{\textcolor{blue}{\scriptstyle\blacktriangledown \mathbf{0.412}}}$&$0.002_{\textcolor{BrickRed}{\scriptstyle\blacktriangle \mathbf{0.002}}}$&$0.102_{\textcolor{blue}{\scriptstyle\blacktriangledown \mathbf{0.245}}}$\\
&$\mathit{75th}$&$0.212_{\textcolor{BrickRed}{\scriptstyle\blacktriangle \mathbf{0.212}}}$&$0.668_{\textcolor{blue}{\scriptstyle\blacktriangledown \mathbf{0.326}}}$&$0.004_{\textcolor{BrickRed}{\scriptstyle\blacktriangle \mathbf{0.004}}}$&$0.173_{\textcolor{blue}{\scriptstyle\blacktriangledown \mathbf{0.266}}}$\\
&$\mathit{max}$&$0.236_{\textcolor{BrickRed}{\scriptstyle\blacktriangle \mathbf{0.236}}}$&$0.737_{\textcolor{blue}{\scriptstyle\blacktriangledown \mathbf{0.259}}}$&$0.030_{\textcolor{BrickRed}{\scriptstyle\blacktriangle \mathbf{0.030}}}$&$0.602_{\textcolor{blue}{\scriptstyle\blacktriangledown \mathbf{0.270}}}$\\
\hline
\multirow{5}{*}{\textbf{$\mathcal{P}$(*,T,*)}}&$\mathit{min}$&$0.000_{\textcolor{blue}{\scriptstyle\blacktriangledown \mathbf{0.000}}}$&$0.419_{\textcolor{blue}{\scriptstyle\blacktriangledown \mathbf{0.567}}}$&$0.000_{\textcolor{blue}{\scriptstyle\blacktriangledown \mathbf{0.000}}}$&$0.006_{\textcolor{blue}{\scriptstyle\blacktriangledown \mathbf{0.014}}}$\\
&$\mathit{25th}$&$0.039_{\textcolor{BrickRed}{\scriptstyle\blacktriangle \mathbf{0.039}}}$&$0.419_{\textcolor{blue}{\scriptstyle\blacktriangledown \mathbf{0.567}}}$&$0.000_{\textcolor{BrickRed}{\scriptstyle\blacktriangle \mathbf{0.000}}}$&$0.066_{\textcolor{blue}{\scriptstyle\blacktriangledown \mathbf{0.175}}}$\\
&$\mathit{50th}$&$0.155_{\textcolor{BrickRed}{\scriptstyle\blacktriangle \mathbf{0.155}}}$&$0.419_{\textcolor{blue}{\scriptstyle\blacktriangledown \mathbf{0.567}}}$&$0.000_{\textcolor{BrickRed}{\scriptstyle\blacktriangle \mathbf{0.000}}}$&$0.097_{\textcolor{blue}{\scriptstyle\blacktriangledown \mathbf{0.228}}}$\\
&$\mathit{75th}$&$0.228_{\textcolor{BrickRed}{\scriptstyle\blacktriangle \mathbf{0.228}}}$&$0.419_{\textcolor{blue}{\scriptstyle\blacktriangledown \mathbf{0.567}}}$&$0.000_{\textcolor{BrickRed}{\scriptstyle\blacktriangle \mathbf{0.000}}}$&$0.137_{\textcolor{blue}{\scriptstyle\blacktriangledown \mathbf{0.312}}}$\\
&$\mathit{max}$&$0.309_{\textcolor{BrickRed}{\scriptstyle\blacktriangle \mathbf{0.309}}}$&$0.419_{\textcolor{blue}{\scriptstyle\blacktriangledown \mathbf{0.567}}}$&$0.154_{\textcolor{BrickRed}{\scriptstyle\blacktriangle \mathbf{0.154}}}$&$0.419_{\textcolor{blue}{\scriptstyle\blacktriangledown \mathbf{0.567}}}$\\
\hline
\multirow{5}{*}{\textbf{$\mathcal{P}$(V,*,*)}}&$\mathit{min}$&$0.000_{\textcolor{blue}{\scriptstyle\blacktriangledown \mathbf{0.000}}}$&$0.000_{\textcolor{blue}{\scriptstyle\blacktriangledown \mathbf{0.000}}}$&$0.000_{\textcolor{blue}{\scriptstyle\blacktriangledown \mathbf{0.000}}}$&$0.000_{\textcolor{blue}{\scriptstyle\blacktriangledown \mathbf{0.000}}}$\\
&$\mathit{25th}$&$0.007_{\textcolor{BrickRed}{\scriptstyle\blacktriangle \mathbf{0.007}}}$&$0.001_{\textcolor{blue}{\scriptstyle\blacktriangledown \mathbf{0.992}}}$&$0.000_{\textcolor{BrickRed}{\scriptstyle\blacktriangle \mathbf{0.000}}}$&$0.001_{\textcolor{blue}{\scriptstyle\blacktriangledown \mathbf{0.302}}}$\\
&$\mathit{50th}$&$0.035_{\textcolor{BrickRed}{\scriptstyle\blacktriangle \mathbf{0.035}}}$&$0.582_{\textcolor{blue}{\scriptstyle\blacktriangledown \mathbf{0.412}}}$&$0.000_{\textcolor{BrickRed}{\scriptstyle\blacktriangle \mathbf{0.000}}}$&$0.087_{\textcolor{blue}{\scriptstyle\blacktriangledown \mathbf{0.268}}}$\\
&$\mathit{75th}$&$0.321_{\textcolor{BrickRed}{\scriptstyle\blacktriangle \mathbf{0.321}}}$&$0.668_{\textcolor{blue}{\scriptstyle\blacktriangledown \mathbf{0.326}}}$&$0.005_{\textcolor{BrickRed}{\scriptstyle\blacktriangle \mathbf{0.005}}}$&$0.179_{\textcolor{blue}{\scriptstyle\blacktriangledown \mathbf{0.234}}}$\\
&$\mathit{max}$&$0.720_{\textcolor{BrickRed}{\scriptstyle\blacktriangle \mathbf{0.720}}}$&$0.737_{\textcolor{blue}{\scriptstyle\blacktriangledown \mathbf{0.259}}}$&$0.400_{\textcolor{BrickRed}{\scriptstyle\blacktriangle \mathbf{0.400}}}$&$0.493_{\textcolor{BrickRed}{\scriptstyle\blacktriangle \mathbf{0.030}}}$\\
\hline
\textbf{$\mathcal{P}$(*,*,*)}&---&$0.141_{\textcolor{BrickRed}{\scriptstyle\blacktriangle \mathbf{0.141}}}$&$0.419_{\textcolor{blue}{\scriptstyle\blacktriangledown \mathbf{0.567}}}$&$0.002_{\textcolor{BrickRed}{\scriptstyle\blacktriangle \mathbf{0.002}}}$&$0.106_{\textcolor{blue}{\scriptstyle\blacktriangledown \mathbf{0.248}}}$\\
    \bottomrule
\end{tabu}
\begin{tablenotes}
\scriptsize
\item[1] $\mathit{25th}$, $\mathit{50th}$, and $\mathit{75th}$ represent the respective percentiles.
\item[2] St., Di., Sub., and Ex. stand for ``stealthy'', ``direct'', ``sub-prefix'', and ``exact-prefix''. The difference ($\textcolor{blue}{\scriptstyle\blacktriangledown}$/$\textcolor{BrickRed}{\scriptstyle\blacktriangle}$) is based on the comparison with a no-ROV scenario.
\end{tablenotes}
\end{ThreePartTable}
\vspace{-2mm}
\end{table}

\subsection{Risk Assessment Results}

We now examine the stealthy hijacking risk in the current Internet based on \ours's results, gaining insights into its prevalence, distribution, and underlying topological features.

\noindent\textbf{Overall Risk Level.}
We first assess the overall stealthy hijacking risk posed by the current partial ROV deployment. For any type of BGP hijacking (stealthy or direct, exact- or sub-prefix), we define its overall risk level as the statistical success probability across random ``victim-target-hijacker'' instances.

Our assessment with the 7,275 ROV-enabled ASes reveals an overall risk level of 0.141 for sub-prefix and 0.002 for exact-prefix stealthy hijacking. This suggests that a random stealthy hijacking attempt has over 14\% probability of success. In comparison, with no ROV, both risks are 0. This stark contrast (0.141 versus 0) highlights the substantial increase in stealthy hijacking risk \emph{exclusively} introduced by the current partial ROV deployment. Meanwhile, we observe the positive effect of ROV in reducing direct hijacking risk. With the same deployment, sub-prefix and exact-prefix direct hijacking risks drop to 0.419 and 0.106, with reductions of 57.5\% and 70.1\%, respectively, compared to a no-ROV scenario. These opposing trends, \ie rising stealthy risk and declining direct risk, reflect the double-edged effect of ROV in partial deployment.

\smallskip
\noindent\textit{\underline{Takeaway 1:} While effectively mitigating direct hijacking risk, the current partial ROV deployment significantly amplifies stealthy hijacking risk from 0 to a 14.1\% overall success probability. This risk arises solely due to ROV deployment.}

\begin{figure*}[t]
\setlength{\belowcaptionskip}{-2mm}
    \centering
    \includegraphics[width=.95\linewidth]{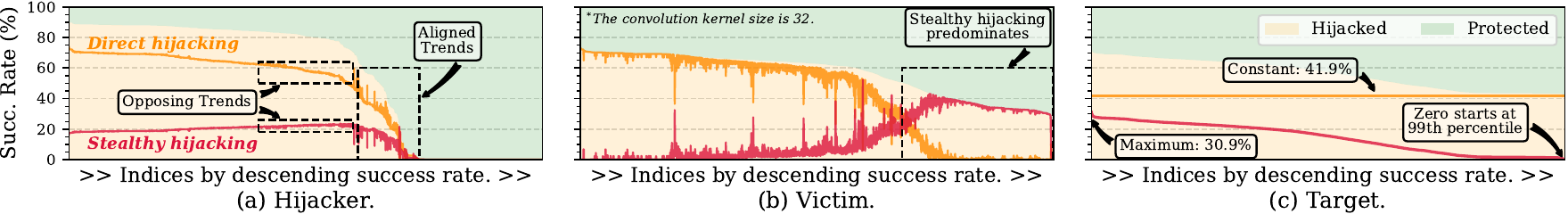}
    \vspace{-2mm}
    \caption{Hijacking success probabilities of ASes in different roles.}
    \label{fig:distribution-over-ases}
\end{figure*}

\smallskip
\noindent\textbf{Aggregated Risk Level.}
Since each hijacking instance involves three distinct entities, \ie the victim, the target, and the hijacker, we aggregate instances by specific entity combinations to assess the risk under certain conditions. For example, aggregating by the ``target-hijacker'' pair creates a set of groups where each comprises all instances with the same target and hijacker but varying victims. Assessing the risk level within each group gives the probability that a \emph{specific} hijacker can successfully hijack a \emph{specific} target on any \emph{random} victim. We refer to this probability as the \emph{aggregated risk level} of the corresponding group. Aggregated risk levels across all such groups, \eg all unique ``target-hijacker'' pairs in this example, collectively form a probability distribution.

We represent an entity combination as a 3-tuple, where absent entities are denoted by * and present entities by their initials; for example, (\scalebox{0.95}{*,T,H}) represents aggregation by the ``target-hijacker'' pair. For each type of BGP hijacking, we assess its aggregated risk level by all entity combinations except (\scalebox{0.95}{V,T,H}), since aggregating by individual instances provides little statistical insight. We use $\mathcal{P}$($\cdot$) to denote the corresponding probability distribution. Notably, (\scalebox{0.95}{*,*,*}) represents no aggregation, so $\mathcal{P}$(\scalebox{0.95}{*,*,*}) reduces to a single value, \ie the overall risk level presented in previous analysis.

Figure~\ref{fig:aggregated-risk-levels} shows the cumulative density function (CDF) of the aggregated risk levels, while Table~\ref{tab:risk-dissection} presents key statistics and highlights differences compared to the no-ROV baseline. In general, the stealthy hijacking risk tends to concentrate on a few specific pairs when aggregated by two entities (see $\mathcal{P}$(\scalebox{0.95}{*,T,H}), $\mathcal{P}$(\scalebox{0.95}{V,*,H}), and $\mathcal{P}$(\scalebox{0.95}{V,T,*})), but spread more evenly across ASes when aggregated by a single entity (see $\mathcal{P}$(\scalebox{0.95}{*,*,H}), $\mathcal{P}$(\scalebox{0.95}{*,T,*}), and $\mathcal{P}$(\scalebox{0.95}{V,*,*})). For example, $\mathcal{P}$(\scalebox{0.95}{*,T,H}), the success probability distribution of a specific hijacker hijacking a specific target, has a median of 0.033 and a maximum of 0.995 under sub-prefix stealthy hijacking. In contrast, $\mathcal{P}$(\scalebox{0.95}{*,*,H}), the success probability distribution of a specific hijacker hijacking any random route, has a higher median of 0.188 but a much lower maximum of 0.236. This disparity reflects the difference in hijacking capacity between \emph{targeted} and \emph{non-targeted} sub-prefix stealthy hijacking: the targeted type enables a few hijackers to achieve near-certain success, while the non-targeted yields only moderate success probability across all hijackers. A similar pattern is observed in exact-prefix stealthy hijacking, but not in direct hijacking, either exact-prefix or sub-prefix. Again, the double-edged effect of partial ROV deployment is evident, as most statistics of stealthy hijacking risk show an increase (highlighted in red in Table~\ref{tab:risk-dissection}), while those of direct hijacking risk show a decrease (highlighted in blue). The only exception is the slight increase in the maximum value for exact-prefix direct hijacking, due to the increase of variance in probability distributions caused by random tie-breaking between competing routes for the exact prefixes.

\smallskip
\noindent\textit{\underline{Takeaway 2:} Targeted stealthy hijacking achieves near-certain success on specific AS pairs (up to 99.5\%), while non-targeted stealthy hijacking distributes risk more evenly across ASes (with a maximum of 23.6\%). In contrast, direct hijacking does not exhibit these patterns.}

\smallskip
\noindent\textbf{Distribution across ASes.}
We further examine how hijacking risk distributes across ASes in different roles by assessing \emph{role-specific} hijacking success probabilities, \ie $\mathcal{P}$(\scalebox{0.95}{V,*,*}), $\mathcal{P}$(\scalebox{0.95}{*,T,*}), and $\mathcal{P}$(\scalebox{0.95}{*,*,H}). Interpretation of hijacking success probability differs by role: for a hijacker, it reflects its capability to hijack routes globally, while for a victim or target, it reflects its exposure to hijacking threats from any hijacker.

We now focus exclusively on sub-prefix hijacking, as it is generally more impactful. Figure~\ref{fig:distribution-over-ases}(a)-(c) present the hijacking success probabilities for ASes acting as hijackers, victims, and targets, respectively, under both stealthy (red curve) and direct (orange curve) hijacking. In each figure, ASes are sorted along the X-axis in descending order of their overall hijacking success probability, \ie the sum of both stealthy and direct hijacking success probabilities. To improve readability, the curves in Figure~\ref{fig:distribution-over-ases}(b) are smoothed using a convolution kernel of size 32. We observe that over 50\% of ASes, as attackers, can hijack over 80\% of global routes despite ROV deployment; even the most protected 1\% of ASes face considerable risk, with about 30\% of their routes vulnerable as victims and 43\% as targets. This highlights the limited protection ROV provides against BGP hijacking under its current deployment.

In Figure~\ref{fig:distribution-over-ases}(a), as the overall success probability decreases along the X-axis, stealthy hijacking shows an upward trend, in contrast to the steady decline of direct hijacking. However, for a few ASes (as hijackers) where ROV filtering is highly effective, both direct and stealthy hijacking's success probabilities align and drop towards zero. The opposing trend reveals a key tradeoff introduced by partial ROV deployment: while ROV blocks malicious routes, it also limits the route visibility of benign ASes, inadvertently exacerbating stealthy hijacking risk. Yet, the convergence of stealthy and direct hijacking's success probabilities at the tail end suggests diminishing returns, \ie once ROV restrictions on attackers become sufficiently strong, stealthy hijacking also loses its advantage.

We observe a similar opposing trend in Figure~\ref{fig:distribution-over-ases}(b), where 32\% of ASes, as victims, face even higher stealthy hijacking risk than direct hijacking, with success probability reaching up to 72\% in extreme cases. Moreover, stealthy hijacking's success probabilities exhibit considerable variability, as indicated by noticeable spikes in the curves, suggesting that the stealthy hijacking risk faced by a victim is more case-specific and less predictable. In contrast, Figure~\ref{fig:distribution-over-ases}(c) shows that stealthy hijacking risk distributes more evenly across ASes as targets, affecting over 99\% of them with a maximum success probability of 30.9\%. Notably, direct hijacking's success probabilities remain constant at 41.9\%, as sub-prefix direct hijacking depends solely on whether victims accept malicious routes from hijackers, irrespective of the targets.

\smallskip
\noindent\textit{\underline{Takeaway 3:} While stealthy hijacking risk mostly opposes the overall risk trend across ASes, its diminishing gain is eventually suppressed as ROV's restrictions on attackers prevail. Besides, the risk is more case-specific across victims but more evenly distributed across targets.}

\begin{figure}[t]
\setlength{\belowcaptionskip}{-2mm}
    \centering
    \includegraphics[width=\linewidth]{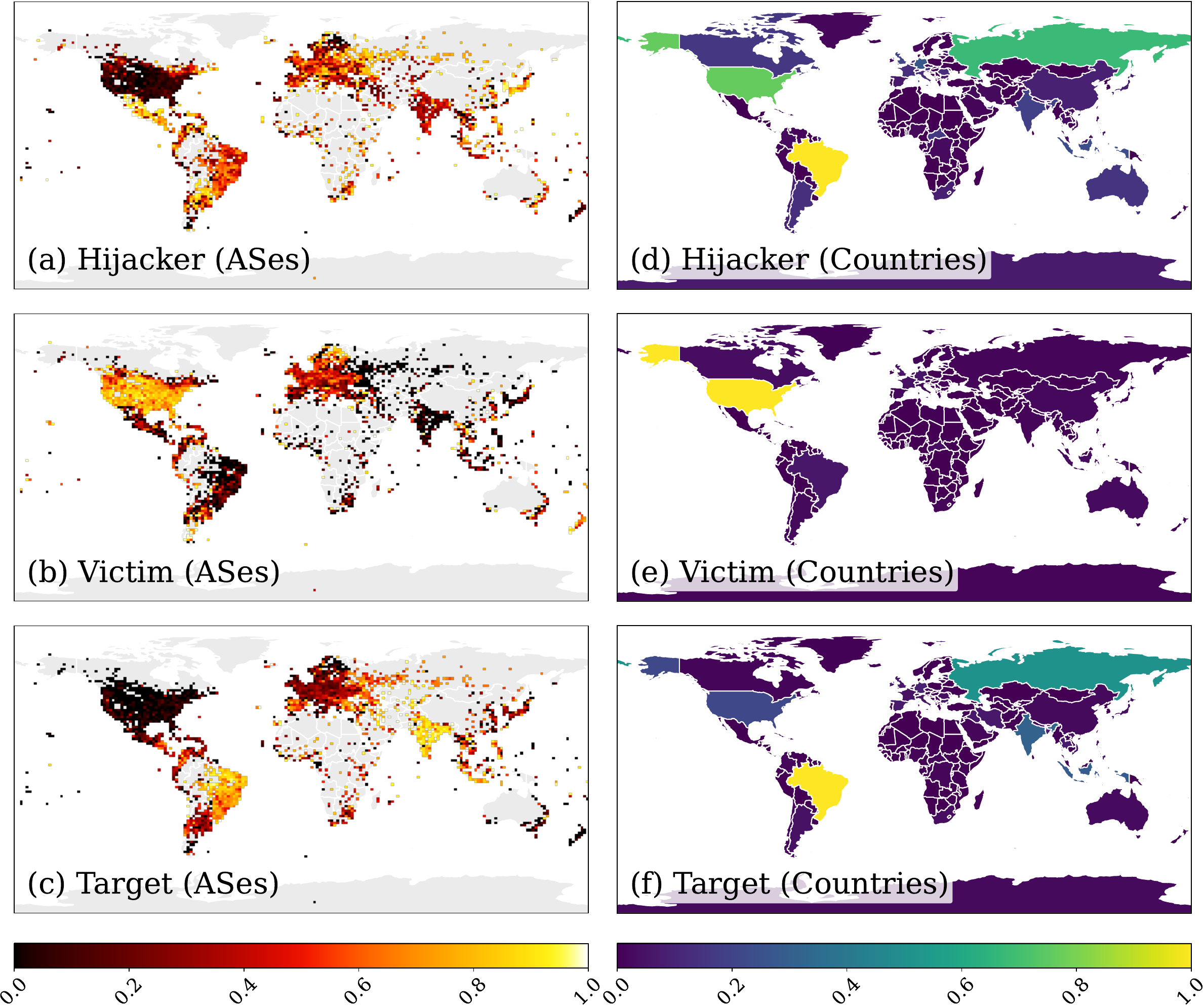}
    \caption{Geographic distribution of stealthy hijacking risk.}
    \label{fig:distribution-over-geo}
\end{figure}

\smallskip
\noindent\textbf{Distribution across Geolocations.}
We now look into the geographic distribution of stealthy hijacking risk, using the MaxMind GeoLite2 dataset~\cite{geolite2} for AS geolocation. Figure~\ref{fig:distribution-over-geo}(a)-(c) map all ASes globally, shaded by the ratio of hijacking instances involving them as hijackers, victims, and targets, respectively. We observe that the most capable potential hijackers cluster in Europe, South America (especially Brazil), and North America (notably Mexico), while victims are mainly located in North America, with the US being the most affected. Meanwhile, the most targeted ASes concentrate in South America (particularly Brazil) and South Asia (notably India). These observations likely stem from differences in regional Internet connectivity. Figure~\ref{fig:distribution-over-ases}(d)-(f) further show these role-specific ratios averaged over each country's total hijacking instances, capturing ASes' relative role tendency per country. We emphasize that this metric reflects statistical role-country correlation, rather than actual or intentional behavior of any country. Besides earlier observations, we find that ASes in the US and Russia are more likely (and capable) to act as hijackers if involved in stealthy hijacking, while both countries also exhibit a considerable fraction of ASes at risk as targets.

\smallskip
\noindent\textit{\underline{Takeaway 4:} ASes most effective in launching stealthy hijacking are mainly in Europe, South America, and North America; victim-prone ASes are mainly in North America; and target-prone ASes are mainly in South America and South Asia.}

\begin{figure}[t]
\setlength{\belowcaptionskip}{-2mm}
    \centering
    \includegraphics[width=.95\linewidth]{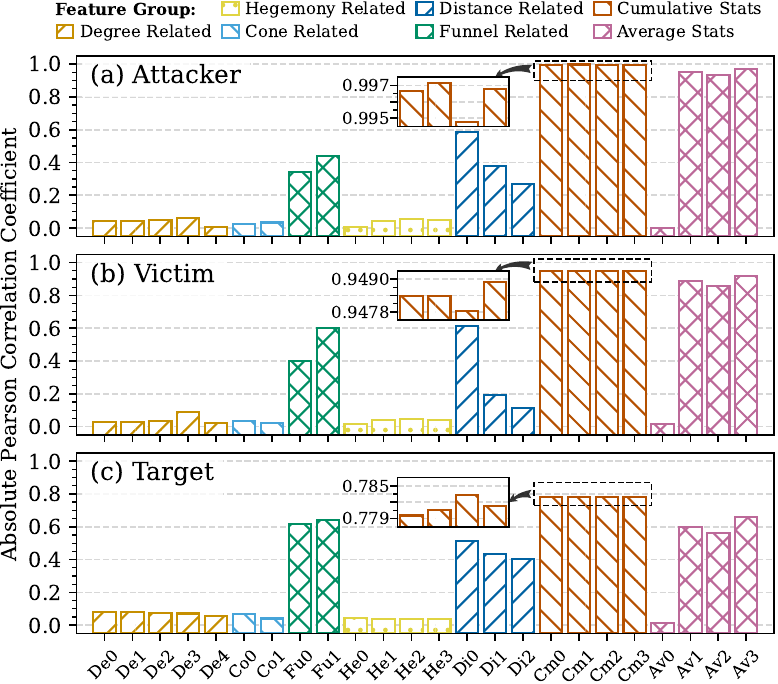}
    \caption{Feature correlation with stealthy hijacking risk.}
    \label{fig:overall-correlation}
\end{figure}

{\renewcommand{\arraystretch}{0.9}
\begin{table*}[t]
\caption{Topological features used to analyze factors influencing stealthy hijacking risk.}
\label{tab:topological-features}
\scriptsize
\centering
\resizebox{\linewidth}{!}{
    \begin{tabular}{lll}
        \toprule
        \textbf{Feature Group} & \textbf{Feature Name} & \textbf{Description}  \\
        \midrule
        \multirow{5}{*}{Degree Related} & De0 (Node Degree)  & Number of edges (AS relationships) connected to the node (AS). \\
        & De1 (Out Degree) & Number of outbound edges (P2C/P2P relationships) from the node (AS). \\
        & De2 (In Degree) & Number of inbound edges (C2P/P2P relationships) to the node (AS). \\
        & De3 (Provider Degree) & Number of direct providers of the AS.\\
        & De4 (Customer Degree) & Number of direct customers of the AS.\\
        \midrule
        \multirow{2}{*}{Cone Related} & Co0 (Customer Cone Size)  & Number of direct or indirect customers of the AS. \\
        & Co1 (ROV Cone Size) & Number of direct or indirect ROV-enabled customers of the AS. \\
        \midrule
        \multirow{2}{*}{Funnel Related} & Fu0 (Provider Funnel Size)  & Number of direct or indirect providers of the AS. \\
        & Fu1 (ROV Funnel Size) & Number of direct or indirect ROV-enabled providers of the AS. \\
        \midrule
        \multirow{4}{*}{Hegemony Related} & He0 (AS Hegemony) & Ratio of routes that traverse the AS. \\
        & He1 (ROV Hegemony) & Ratio of routes traversing both ROV-enabled ASes and the given AS. \\
        & He2 (Pre-ROV Hegemony) & Ratio of routes traversing ROV-enabled ASes before the given AS. \\
        & He3 (Post-ROV Hegemony) & Ratio of routes traversing ROV-enabled ASes after the given AS. \\
        \midrule
        \multirow{3}{*}{Distance Related} & Di0 (Minimum ROV Distance) & Minimum path length from the AS to any reachable ROV-enabled AS. \\
        & Di2 (Maximum ROV Distance) & Maximum path length from the AS to any reachable ROV-enabled AS. \\
        & Di3 (Average ROV Distance) & Average path length from the AS to all its reachable ROV-enabled ASes \\
        \midrule
        \multirow{4}{*}{Cumulative Statistics} & Cm0 (Cumulative AS Hegemony) & Sum of He0 (AS Hegemony) over ASes reachable by the given AS within malicious reach. \\
        & Cm1 (Cumulative ROV Hegemony) & Sum of He1 (ROV Hegemony) over ASes reachable by the given AS within malicious reach. \\
        & Cm2 (Cumulative Pre-ROV Hegemony) & Sum of He2 (Pre-ROV Hegemony) over ASes reachable by the given AS within malicious reach. \\
        & Cm3 (Cumulative Post-ROV Hegemony) & Sum of He3 (Post-ROV Hegemony) over ASes reachable by the given AS within malicious reach. \\
        \midrule
        \multirow{4}{*}{Average Statistics} & Av0 (Average AS Hegemony) & Average of He0 (AS Hegemony) over ASes reachable by the given AS within malicious reach. \\
        & Av1 (Average ROV Hegemony) & Average of He1 (ROV Hegemony) over ASes reachable by the given AS within malicious reach. \\
        & Av2 (Average Pre-ROV Hegemony) & Average of He2 (Pre-ROV Hegemony) over ASes reachable by the given AS within malicious reach. \\
        & Av3 (Average Post-ROV Hegemony) & Average of He3 (Post-ROV Hegemony) over ASes reachable by the given AS within malicious reach. \\
        \bottomrule
    \end{tabular}
}
\end{table*}
}

\smallskip
\noindent\textbf{Influencing Factors.}
To understand how Internet topology affects stealthy hijacking risk, we examine a broad range of topological features and evaluate their statistical correlation with the risk. \morerevised{As shown in Table~\ref{tab:topological-features}}, we consider 24 features, including both basic metrics and composite statistics, grouped into seven categories: degree related (De0-De4), cone related (Co0/Co1), funnel related (Fu0/Fu1), hegemony related (He0-He3), distance related (Di0-Di2), cumulative statistics (Cm0-Cm3), and average statistics (Av0-Av3). Degree and distance related features reflect AS connectivity and centrality in the Internet topology~\cite{faloutsos1999power}. Cone and funnel related features indicate the capability of ASes to provide and access transit services, respectively~\cite{luckie2013relationships,giotsas2014inferring,prehn2024kirin}. Hegemony related features measure AS interdependencies in forming routing paths~\cite{fontugne2018thin}. Cumulative and average statistics, derived by aggregating these features over AS neighborhoods, capture broader topological characteristics at various scales.

We use the Pearson Correlation Coefficient (PCC)~\cite{pearson1895vii} to quantify the linear correlation between each feature and stealthy hijacking risk (measured by hijacking success probability). This coefficient, denoted $r\mypart{1}$, ranges from -1 to 1, with values closer to 1 (or -1) indicating stronger positive (or negative) correlation. To capture potential non-linear correlation, given Internet topology's scale-free nature~\cite{faloutsos1999power}, we further apply quadratic regression and compute the PCC between the fitted and observed risk levels. The resulting coefficient, denoted $r\mypart{2}$, reflects the strength of a quadratic fit. For each feature, we report the higher absolute value of $r\mypart{1}$ and $r\mypart{2}$ in Figure~\ref{fig:overall-correlation}. Notably, cumulative statistics exhibit the strongest correlation overall, with $r\mypart{2}$ of Cm1 (Cumulative ROV Hegemony) reaching up to 0.997, 0.948, and 0.782 for hijackers, victims and targets, respectively. Certain funnel related features, \eg Fu1 (ROV Funnel Size), distance related features, \eg Di0 (Minimum ROV Distance), and average statistics, \eg Av1 (Average ROV Hegemony), also exhibit relatively strong correlation with stealthy hijacking risk.

\begin{figure*}[t]
\setlength{\belowcaptionskip}{-3mm}
    \centering
    \includegraphics[width=\linewidth]{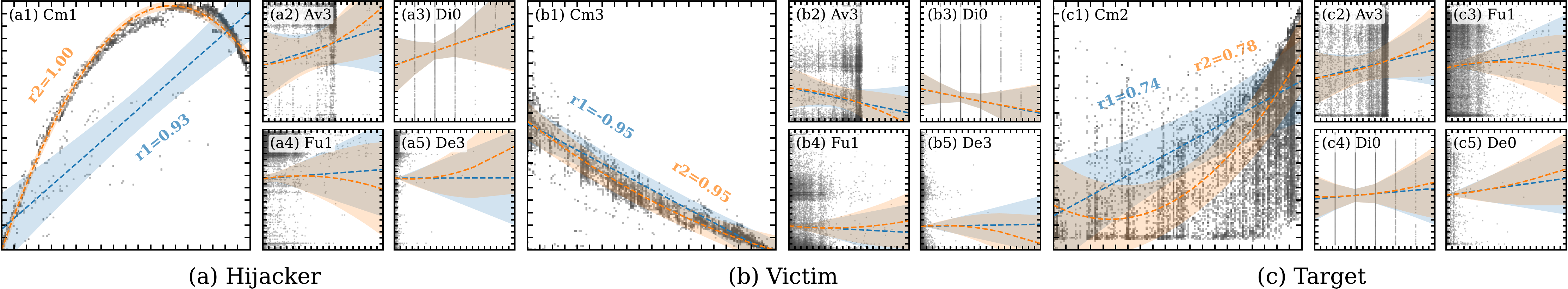}
    \caption{Illustration of statistical correlation between selected features and stealthy hijacking risk.}
    \label{fig:correlation-scatter}
\end{figure*}

To further examine the most correlated features, we plot ASes on a plane, where the X-axis represents the feature value and the Y-axis represents the stealthy hijacking risk level, as shown in Figure~\ref{fig:correlation-scatter}. Fitted values from linear and quadratic regressions are presented with blue lines and orange curves, respectively, with shaded areas indicating estimation errors. The risk levels show a clear quadratic correlation with cumulative statistics (see Figure~\ref{fig:correlation-scatter}(a1), (b1), and (c1)), but less so with other features. Particularly, Cm1 in (a1) measures the fraction of routes traversing any ROV-enabled AS before reaching the given AS's reachable ASes within malicious reach. This statistic reflects the range of potential risk-critical ASes that the given AS as a hijacker can compromise. As this range expands, more ASes become potential targets, yet potential victims become fewer, resulting in the observed parabolic curve. The near-perfect quadratic fit (with $r\mypart{2}$ approaching 1.00) suggests that cumulative statistics are reliable indicators of stealthy hijacking risk. Notably, the opposite concavity of the parabola in (a1) compared to (b1) and (c1) reflects the opposite interpretation of risk levels in terms of different roles, \ie hijacking capability for ASes acting as hijackers, and hijacking susceptibility for those as victims or targets.

\smallskip
\noindent\textit{\underline{Takeaway 5:} Cumulative statistics of AS hegemony show the strongest quadratic correlation with stealthy hijacking risk, making them powerful indicators for predicting risk levels.}

\smallskip
\noindent\textbf{Risk Attribution.}
To better understand how risk-critical ASes and ROV-enabled ASes contribute to the risk, we examine the frequency of stealthy hijacking instances associated with them. We define the risk-critical AS as the first AS along the victim-to-target route that has a route to the hijacker, and the responsible ROV-enabled AS as the first ROV-enabled AS along that path. We attribute all stealthy hijacking instances to 8,323 unique risk-critical ASes and 1,608 unique ROV-enabled ASes. Figure~\ref{fig:risk-attribution} shows the contribution of these ASes to stealthy hijacking risk. The blue curves represent the count of instances associated with each AS, and the orange curves show the cumulative percentage. Dashed lines mark the 80\% cut-off, revealing a pronounced long-tail effect, \ie a vital few account for the majority of the risk. That is, 2.94\% of risk-critical ASes (245 out of 8,323) or 2.24\% of ROV-enabled ASes (36 out of 1,608) are responsible for 80\% of all stealthy hijacking instances. This observation highlights the importance of focusing mitigation efforts on a small set of ASes to effectively reduce the overall risk of stealthy hijacking.

\begin{figure}[t]
\setlength{\belowcaptionskip}{-4mm}
    \centering
    \includegraphics[width=0.95\linewidth]{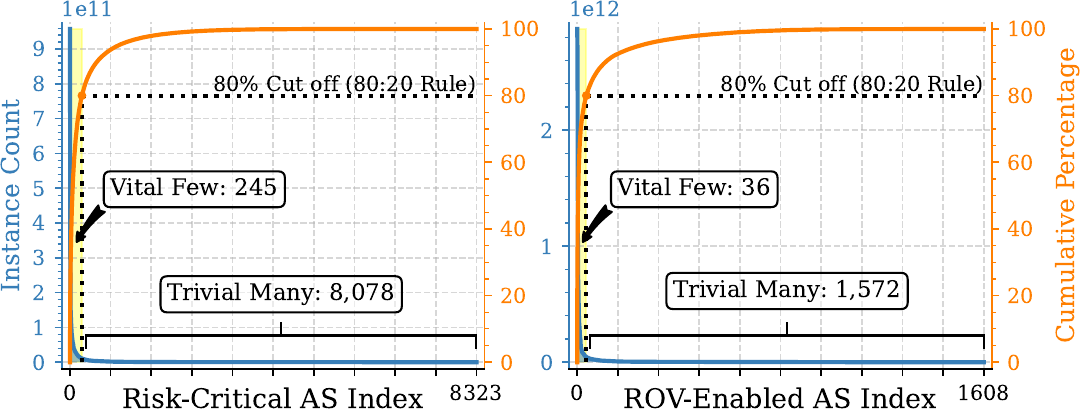}
    \caption{Attribution of stealthy hijacking risk to risk-critical ASes (left) and ROV-enabled ASes (right).}
    \label{fig:risk-attribution}
\end{figure}

\smallskip
\noindent\textit{\underline{Takeaway 6:} A small fraction of risk-critical and ROV-enabled ASes account for the majority of stealthy hijacking risk, calling for focused risk mitigation efforts on these key ASes.}

\smallskip
\noindent\textbf{Evolution Pattern.}
Finally, we analyze how stealthy hijacking risk evolves with increasing ROV deployment, as shown in Figure~\ref{fig:evolution-pattern}. Our preliminary assessment based on data from October 1, 2023 identifies 778 ROV-enabled ASes, which accounts for approximately 1\% of the 75k ASes active at the time. It reports an overall stealthy hijacking success probability (including both sub-prefix and exact-prefix ones) of 0.105. In comparison, data from March 1, 2025 identify 7,275 ROV-enabled ASes (roughly 10\% of all ASes), with the risk rising to 0.145, as reported earlier. This reflects an increase in overall risk alongside real-world ROV deployment. Moreover, we simulate future ROV deployment by progressively relaxing the threshold for identifying ASes as ROV-enabled, and the results suggest that the risk begins to decline as deployment increases. Specifically, a 10\% additional increase in the ROV deployment rate reduces the hijacking success probability to 0.101, and a rate exceeding 40\% reduces the risk to near zero.

\smallskip
\noindent\textit{\underline{Takeaway 7:} The stealthy BGP hijacking risk shows an overall ``rise-then-decline'' evolution pattern along ROV deployment, and we may now be entering the declining phase.}

\begin{figure}[t]
\setlength{\belowcaptionskip}{-2mm}
    \centering
    \includegraphics[width=.9\linewidth]{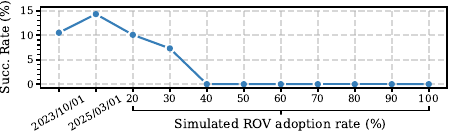}
    \caption{\revised{Stealthy hijacking risk evolution.}}
    \label{fig:evolution-pattern}
\end{figure}
\section{Framework Performance Evaluation}
\label{sec:performance-evaluation}

\subsection{Risk Discovery Effectiveness}

\begin{figure}[t]
\setlength{\belowcaptionskip}{-4mm}
    \centering
    \includegraphics[width=.9\linewidth]{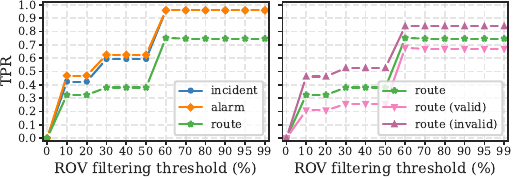}
    \caption{True positive rate of our framework.}
    \label{fig:accuracy-threshold}
\end{figure}

\begin{table}[t]
\newcommand*\emptycirc[1][1ex]{\tikz\draw (0,0) circle (#1);}
\newcommand*\fullcirc[1][1ex]{\tikz\fill (0,0) circle (#1);}
\newcommand\sample{{\protect\tikz[x=3ex,y=1.5ex] \protect\fill [gray!20] (0,0) rectangle (1,1);}\ }
\renewcommand{\arraystretch}{1.0}
    \caption{True positive rate under ablation of ROV sources.}
    \label{tab:rov-ablation}
    \scriptsize
    \centering
\begin{ThreePartTable}
    \begin{tabular}{ccccccccr}
    \toprule
    \multicolumn{3}{c}{\textbf{Sources\textsuperscript{1}}} & \multicolumn{5}{c}{\textbf{True Positive Rate}\textsuperscript{2}} & \multirow{2}[2]{*}{\textbf{\#ASes}} \\
    \cmidrule(lr){1-3}
    \cmidrule(lr){4-8}
    \textbf{A} & \textbf{R} & \textbf{C} & \textbf{Incident} & \textbf{Alarm} & \textbf{Route} & \textbf{Rt.V.} & \textbf{Rt.Iv.} & \\
    \midrule
    \rowcolor{gray!20}
\fullcirc[0.5ex]&\fullcirc[0.5ex]&\fullcirc[0.5ex]&\textbf{0.9591}&\textbf{0.9597}&\textbf{0.7521}&\textbf{0.6782}&0.8400&7,275 \\
\emptycirc[0.5ex]&\fullcirc[0.5ex]&\fullcirc[0.5ex]&0.7862&0.8012&0.6290&0.4519&0.8400&6,725 \\
\fullcirc[0.5ex]&\emptycirc[0.5ex]&\fullcirc[0.5ex]&0.2767&0.2882&0.4201&0.0971&0.8048&2,668 \\
\fullcirc[0.5ex]&\fullcirc[0.5ex]&\emptycirc[0.5ex]&0.8019&0.8184&0.5969&0.3902&0.8431&7,209 \\
\fullcirc[0.5ex]&\emptycirc[0.5ex]&\emptycirc[0.5ex]&0.2767&0.2882&0.4201&0.0971&0.8048&2,575 \\
\emptycirc[0.5ex]&\fullcirc[0.5ex]&\emptycirc[0.5ex]&0.6164&0.6484&0.5202&0.2492&0.8431&6,655 \\
\emptycirc[0.5ex]&\emptycirc[0.5ex]&\fullcirc[0.5ex]&0.0629&0.0576&0.4656&0.0169&\textbf{1.0000}&165 \\
\emptycirc[0.5ex]&\emptycirc[0.5ex]&\emptycirc[0.5ex]&0.0618&0.0720&0.4449&0.0166&0.9551&1,000 \\
    \bottomrule
    \end{tabular}
\begin{tablenotes}
    \scriptsize
    \item[1] A, R, and C denote APNIC, RoVista, and Cloudflare, resp. The row in \sample is ours. The row without any source randomly selects 1,000 ASes as ROV-enabled. 
    \item[2] Rt.V. and Rt.Iv. indicate valid and invalid routes, resp. Highest values are in bold. 
\end{tablenotes}   
\end{ThreePartTable}
\end{table}

\begin{figure}[t]
\setlength{\belowcaptionskip}{-4mm}
    \centering
    \includegraphics[width=.9\linewidth]{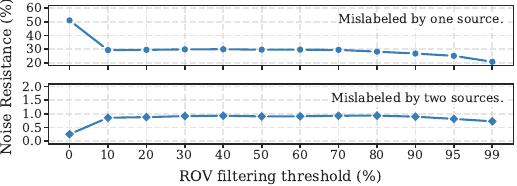}
    \caption{Noise resistance of multi-source ROV input.}
    \label{fig:noise-resistance}
\end{figure}

When evaluating the effectiveness of \ours in risk discovery, we use the curated real-world incident dataset described in \S\ref{sec:empirical-study} as the ground truth. 
The dataset contains 318 high-confidence stealthy hijacking incidents captured in the wild. Each incident consists of one or more related alarms triggered by the same misbehaving origin, and each alarm includes one or more pairs of routes indicating stealthy hijacking to the same prefix. In total, the dataset contains 347 alarms and 2,178 routes. As the best-effort ground truth contains only positive instances\footnote{It is infeasible to establish negative ground, because it is inherently difficult to find a case and tell whether it has not occurred or cannot occur.}, we report the True Positive Rate (TPR) as the main metric, reflecting how effectively \ours infers real-world instances.

We begin by measuring \ours's TPR at the route level. As per the strict heuristics in \S\ref{sec:empirical-study}, each ground truth route pair includes one RPKI-valid and one RPKI-invalid route. For both routes, we extract the vantage point and origin AS to look up the corresponding route inferred by \ours. If the inferred route pair also satisfies the strict heuristics, meaning that \ours's inference matches the ground truth, we consider both routes are effectively inferred. Route-level TPR is defined as the fraction of such effectively inferred routes. We further measure alarm-level TPR as the fraction of alarms with at least one route pair effectively inferred, and incident-level TPR as the fraction of incidents containing at least one such alarm.

We present \ours's TPR across varying ROV filtering thresholds in Figure~\ref{fig:accuracy-threshold} (left). As the threshold increases, fewer but more confident ROV-enabled ASes are selected, and the TPR increases steadily. Beyond the 0.6 threshold, inference remains consistently strong, reaching 0.959, 0.960, and 0.752 at the incident, alarm, and route levels, respectively. This means that over 95\% of these real-world incidents are successfully flagged risky by \ours. Figure~\ref{fig:accuracy-threshold} (right) further breaks down route-level TPR by RPKI validity. As expected, \ours is slightly less effective on RPKI-valid routes due to the stealthiness requirement (condition 3 in \S\ref{sec:heuristics}) that applies exclusively to them. Still, \ours achieves up to 0.678 TPR on these routes. Overall, these results demonstrate the fidelity of \ours's analytical output and inform a recommended threshold range, which justifies our default choice of 0.8.

\subsection{Input Ablation and Robustness}

We examine how incorporating multiple ROV measurement sources improves \ours's effectiveness and robustness. Table~\ref{tab:rov-ablation} reports \ours's TPR under ablation of the three sources, with the ROV filtering threshold fixed at 0.8. In general, excluding any source leads to a noticeable drop in the coverage of ROV-enabled ASes and a corresponding decline in \ours's effectiveness, except that using only the Cloudflare source yields the highest TPR on RPKI-invalid routes. This is because Cloudflare contributes only 165 ROV-enabled ASes, which, according to \ours's assessment, have limited effect in blocking propagation of RPKI-invalid routes. As a result, \ours's inference overestimates observation of RPKI-invalid routes and thus aligns well with the ground truth. However, this setting is impractical and results in just 0.0169 TPR on RPKI-valid routes and 0.0629 at the incident level. After all, we conclude that integrating multiple reliable sources to obtain a more complete view of ROV deployment is crucial to achieve effective stealthy hijacking risk assessment.

We further look into \ours's robustness to noise in the three ROV measurements. We simulate scenarios where one or two sources randomly mislabel an ROV-enabled AS as non-ROV-enabled, and measure the probability that such ASes remain included in the final ROV set. The results are shown in Figure~\ref{fig:noise-resistance}. When all three sources are used, \ours retains a 20\%-50\% probability of preserving an AS despite a mislabel from a single source (see the top panel), and about a 1.0\% probability when two sources mislabel the same AS. By comparison, using only one source result in zero probability of noise resistance, as any mislabel is directly accepted to the final ROV set. These results emphasize the robustness benefit of \ours's design using multiple ROV sources as input.

\subsection{Runtime Overhead}

For comparison, we select BGPsim~\cite{brandt2021optimized}, a highly optimized BGP simulator widely used in prior studies~\cite{hlavacek2022behind,hlavacek2022smart,brandt2021evaluating}. We test \ours and BGPsim on the same Internet topology derived from CAIDA AS relationship data. Figure~\ref{fig:runtime-performance} reports their runtime overhead in generating routes within random AS subsets ranging from 10 to 100 ASes. Each run is repeated 10 times, with error bars indicating the 95\% confidence interval. The top panel compares the runtime of \ours (in a single-thread CPU setting) with BGPsim's, with \ours over 40 times faster in the worst case. The bottom panel evaluates \ours's runtime under different settings (1, 20, and 40 CPU threads, or a single GPU). The settings with 40 CPU threads or a single GPU show the best performance, completing route generation across 75,846 ASes in 5.22 hours with peak memory usage under 20 GiB. In contrast, an exponential fit of BGPsim's overhead estimates 110 days for full route generation. As such, \ours achieves a 500-fold speedup, providing the efficiency necessary for comprehensive Internet-scale stealthy hijacking risk assessment.

\begin{figure}[t]
\setlength{\belowcaptionskip}{-4mm}
    \centering
    \includegraphics[width=.85\linewidth]{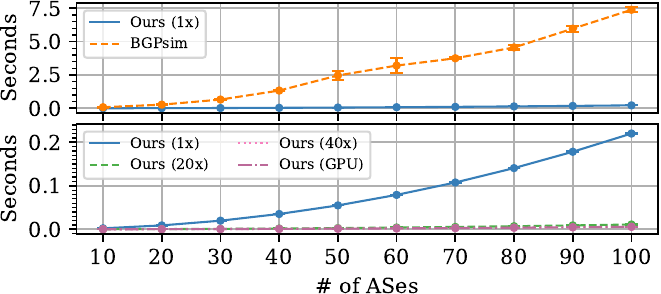}
    \caption{Runtime overhead of route inference.}
    \label{fig:runtime-performance}
\end{figure}
\section{Discussion}

\noindent\textbf{Complex Routing Policy.}
Our framework infers BGP routes based on the established Gao-Rexford model~\cite{gao2001stable,brandt2021optimized,hlavacek2022behind,hlavacek2022smart}. While it cannot capture all real-world nuances, \eg hybrid or partial-transit relationships~\cite{giotsas2014inferring} and selective ROV filtering~\cite{cloudflare}, such cases are rare, \eg only seven ASes exhibit selective filtering~\cite{cloudflare}. Our evaluation on real incidents confirms the framework's reliability despite these complexities. The design is also extensible: additional bytes can encode nuanced relationships, and selective filtering can be modeled by pruning certain links. We leave these to future work.

\noindent\textbf{Mitigation Strategies.}
Increasing ROV deployment remains central to mitigating stealthy hijacking, as a sufficiently high adoption rate can suppress the risk (see Takeaway 3). ROV++\cite{morillo2021rov++}, which extends ROV with proactive rerouting and blackholing upon detecting invalid routes, can also mitigate the threat. Besides, detecting routing table discrepancies across vantage points, as demonstrated in \S\ref{sec:empirical-study}, enables timely alert of ongoing incidents. Operational practices such as announcing /24 IPv4 or /48 IPv6 prefixes also significantly limits the sub-prefix risk. Furthermore, information-sharing platforms, \eg mailing-lists where ROV-adopters publish digests of dropped routes, help potential victims identify risk-critical ASes. A similar routing policy was described in Cisco's patent~\cite{heitz2018poison}.

\noindent\textbf{Incident Intent.}
Identifying intent behind BGP incidents, \ie whether malicious or not, is inherently challenging. We attempt systematic attribution methods (see Table~\ref{tab:incident-tags}) and manual investigation (see Appendix~\ref{sec:appendix:case-study}). While these help narrow down likely causes, definitive confirmation requires engagement with network operators, which we attempted in selected cases. Cross-layer analysis with external intelligence may offer further insights, which we leave as future work. However, regardless of intent, the operational impact of stealthy hijacking remains significant. Unintentional misconfigurations can lead to the same consequences as deliberate attacks.
\section{Related Work}
\label{sec:related-work}

\noindent\textbf{Stealthy Hijacking Analysis.}
Prior work mainly focuses on non-ROV-related stealthy hijacking based on short-lived routes~\cite{vervier2015mind}, AS-path poisoning~\cite{milolidakis2023effectiveness}, or BGP communities~\cite{birge2025global}. To our knowledge, ROV++~\cite{morillo2021rov++} is the only recent work addressing ROV-related stealthy hijacking. It extends ROV with proactive rerouting and blackholing to mitigate the threat, showing promising benefits at early adoption. Particularly, ROV++ adopters can effectively secure all routes traversing them. If AS B in Figure~\ref{fig:problem-statement} deploys ROV++, it drops the invalid route and seeks an alternative route in its RIB that avoids ASes in the invalid route (\ie AS C, F, and G), thus avoiding the risk-critical ASes. If such route is not available, it blocks traffic to the prefix to prevent hijacking. However, in certain mixed-deployment scenarios, \eg when AS B deploys ROV and AS A deploys ROV++, stealthy hijacking is still possible because the malicious route never reaches AS A, thus failing to trigger the countermeasures in ROV++. ROV++ focuses on mitigation strategies and does not provide real-world evidence or systematic assessment of the threat’s prevalence and impact. In contrast, our work aims at real-world stealthy hijacking discovery and systematic risk assessment, addressing this critical gap. Moreover, our risk attribution analysis complements ROV++ by identifying where its deployment would be most effective. For example, the 36 ROV-enabled ASes in Figure~\ref{fig:risk-attribution}, which are mostly Tier-1 ASes and account for 80\% of risk instances, are prime candidates for deploying ROV++.

\noindent\textbf{ROV Deployment Measurement.}
ROV deployment is typically measured by analyzing ASes' data-plane reachability to RPKI-valid and invalid prefixes~\cite{gilad2016we,testart2020filter,hlavacek2018practical,apnic,rodday2021revisiting,cloudflare,li2023rovista,hlavacek2023keep,chen2022rov}. These studies vary in probing sources and prefix selection. For example, RoVista~\cite{li2023rovista} uses IPID side channels; APNIC~\cite{apnic} conducts large-scale probing via its infrastructure; and Cloudflare~\cite{cloudflare} hosts test sites and crowdsources measurements. Despite diverse techniques, each offers only partial coverage due to scalability limits. To improve visibility, we consolidate these three representative sources~\cite{li2023rovista,apnic,cloudflare} for a more comprehensive view of ROV deployment.

\noindent\textbf{BGP Route Inference.}
Existing BGP route inference methods are either simulation-driven~\cite{brandt2021optimized,furuness2023bgpy,caesar2005bgp,bphs,dimitropoulos2006efficient} or data-driven~\cite{li2023realizing,madhyastha2009iplane,li2020probinfer,cunha2016sibyl,tao2015path}. Simulation-driven methods simulate the route exchange using heuristics such as the Gao-Rexford model~\cite{gao2001inferring}. For example, Brandt et al.~\cite{brandt2021optimized} implement BGPsim with bi-directional search to improve performance. Data-driven methods infer AS-level paths from measurements. For example, Cunha et al.~\cite{cunha2016sibyl} design a traceroute-based prediction system. 
While effective at certain scales, these methods cannot generate complete Internet-scale routes efficiently. Our work, by contrast, computes all AS-level routes within a few hours.
\section{Conclusion}

In this paper, we develop effective heuristics to discover stealthy BGP hijacking and conduct the first empirical study to track it in the wild, contributing a curated dataset and a monitoring service. To assess the risk comprehensively, we further design \ours, a framework that integrates multiple sources for accurate ROV deployment, leverages matrix operations to infer Internet-wide routes, and enables systematic risk analysis via a 3-tuple model. It generates Internet-scale routes within hours and achieves 95.9\% incident-level accuracy. Assessing over 8.3 billion routes reveals a 14.1\% success probability for stealthy hijacking, with targeted attacks reaching up to 99.5\%.

\section*{Ethics Considerations}

We carefully consider several ethical aspects to ensure that our study adheres to established ethical standards. Our study only uses publicly available data, and we strictly comply with all terms of use. Our study does not disclose any personally identifiable information or private routing policies beyond what is already publicly available. We do not perform any large-scale active probing or interfere with live routing systems, thereby ensuring no impact on real-world traffic or network stability. The real-world incidents captured during our study are responsibly disclosed to relevant network operators.


\section*{Acknowledgment}
We sincerely thank our Shepherd and all anonymous reviewers for their valuable comments. This work is supported in part by NSFC under Grant 62132011, Grant 62472247, and Grant 62425201. Qi Li is the corresponding author of the paper.


\bibliographystyle{IEEEtran}
\bibliography{ref}

\appendices
\section{Case Study: Stealthy Hijacking on 203.127.225.0/24}
\label{sec:appendix:case-study}

To verify the discovered stealthy hijacking incident targeting 203.127.225.0/24 (depicted in Figure~\ref{fig:case-study}), we manually investigate AS37100's control-plane visibility and data-plane reachability based on first-hand observations from its looking glass\footnote{\url{https://lg.seacomnet.com/}} ``lg-01-ams.nl''. All observations presented below were captured on February 10, 2025.

We begin by examining AS37100's route to the prefix 203.127.0.0/16. As shown in Figure~\ref{fig:looking-glass-1}, running the command ``show ip bgp 203.127.0.0/16'' on the looking glass reveals that AS37100 has a route to the prefix announced by the legitimate origin AS3758, via the AS path 37100 6762 6461 7473 3758. This confirms that the legitimate route is visible to AS37100 on the control plane, aligning with our expectations.

Next, we investigate AS37100's control-plane visibility of the sub-prefix 203.127.225.0/24. As shown in Figure~\ref{fig:looking-glass-2}, executing ``show ip bgp 203.127.225.0/24'' on the looking glass returns no matching routes, indicating that AS37100 does not accept the bogus route announcement from the unauthorized origin AS17894. This corroborates our analysis that the bogus route is not visible to the victim on the control plane.

To obtain the actual data-plane forwarding path, we perform traceroute probing from AS37100 to 203.127.225.0/24. As shown in Figure~\ref{fig:looking-glass-3}, the resulting per-hop data-plane path reveals that the last two hops (line 13 and 14) belong to AS17894. This confirms that traffic from AS37100 to the sub-prefix is indeed diverted to the illegitimate origin, demonstrating actual hijacking at the data-plane level.

Taken together, these observations provide strong evidence of a stealthy BGP hijacking incident, where AS37100’s traffic is misrouted despite its control-plane filtering. However, we emphasize that the intent behind this incident remains uncertain. Given our broader investigation showing that \first AS17894's parent organization, \emph{Innove Communications}, is a subsidiary of \emph{Globe Telecom}~\cite{wikiglobetelecom}, \second AS3758's parent organization, \emph{SingNet}, is operated by \emph{SingTel}~\cite{wikisingtel}, and \third \emph{SingTel} is the principal shareholder of \emph{Globe Telecom}~\cite{globetelecomshareholder}, we suspect that this incident, despite manifesting as stealthy hijacking, is likely the result of overlooked misconfigurations rather than a deliberate attack. As of this writing, we are awaiting confirmation from \emph{Globe Telecom}. Note that this does not diminish the value of the case study, as misconfigurations can cause stealthy hijacking just like intentional attacks.

\begin{figure}[t]
\setlength{\belowcaptionskip}{-2mm}
    \centering
    \includegraphics[width=\linewidth,cfbox=black 0.5pt 0.5pt]{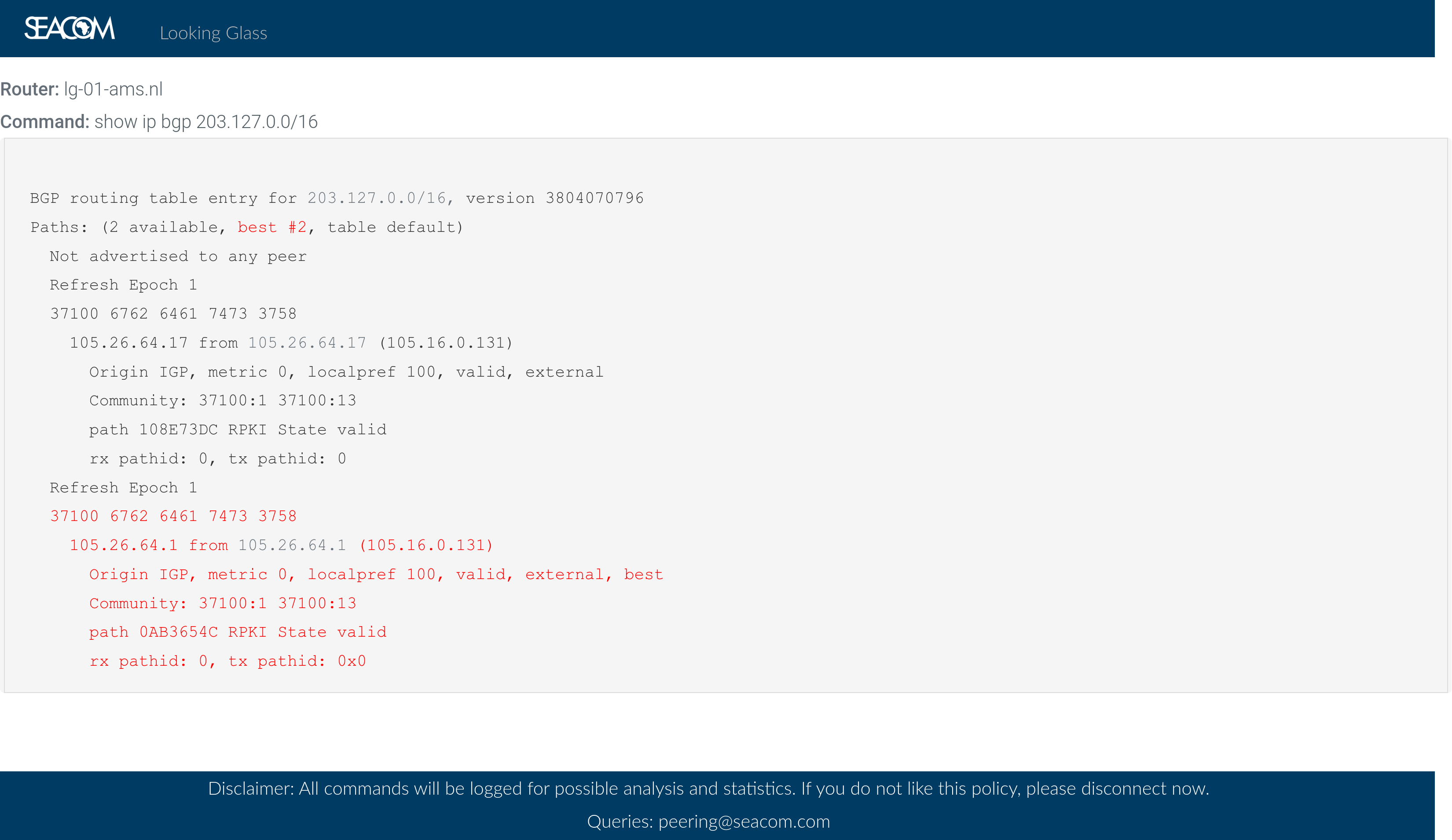}
    \caption{AS37100's routes to 203.127.0.0/16.}
    \label{fig:looking-glass-1}
\end{figure}

\begin{figure}[t]
\setlength{\belowcaptionskip}{-2mm}
    \centering
    \includegraphics[width=\linewidth,cfbox=black 0.5pt 0.5pt]{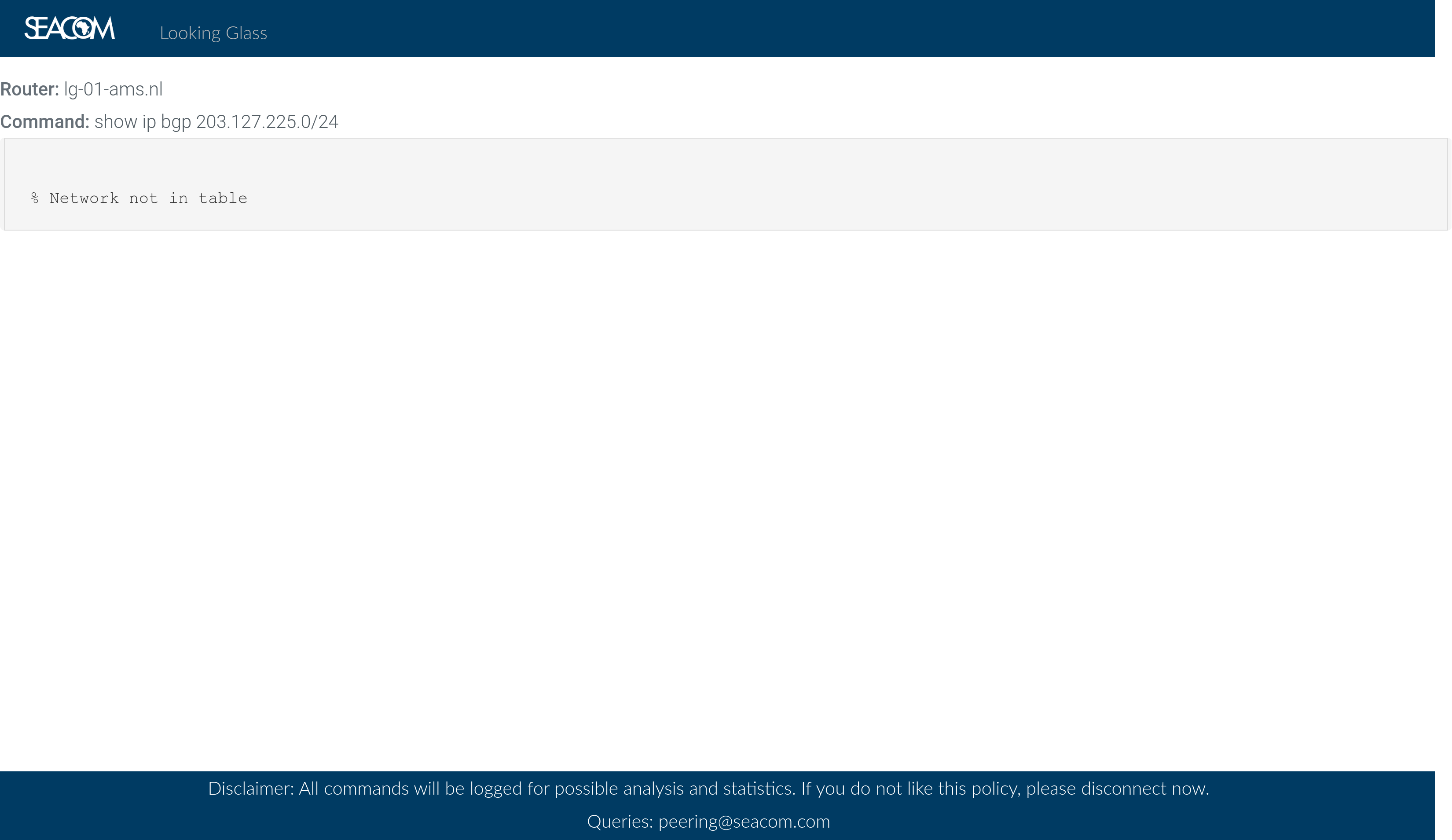}
    \caption{AS37100's routes to 203.127.225.0/24.}
    \label{fig:looking-glass-2}
\end{figure}

\begin{figure}[t]
\setlength{\belowcaptionskip}{-2mm}
    \centering
    \includegraphics[width=\linewidth,cfbox=black 0.5pt 0.5pt]{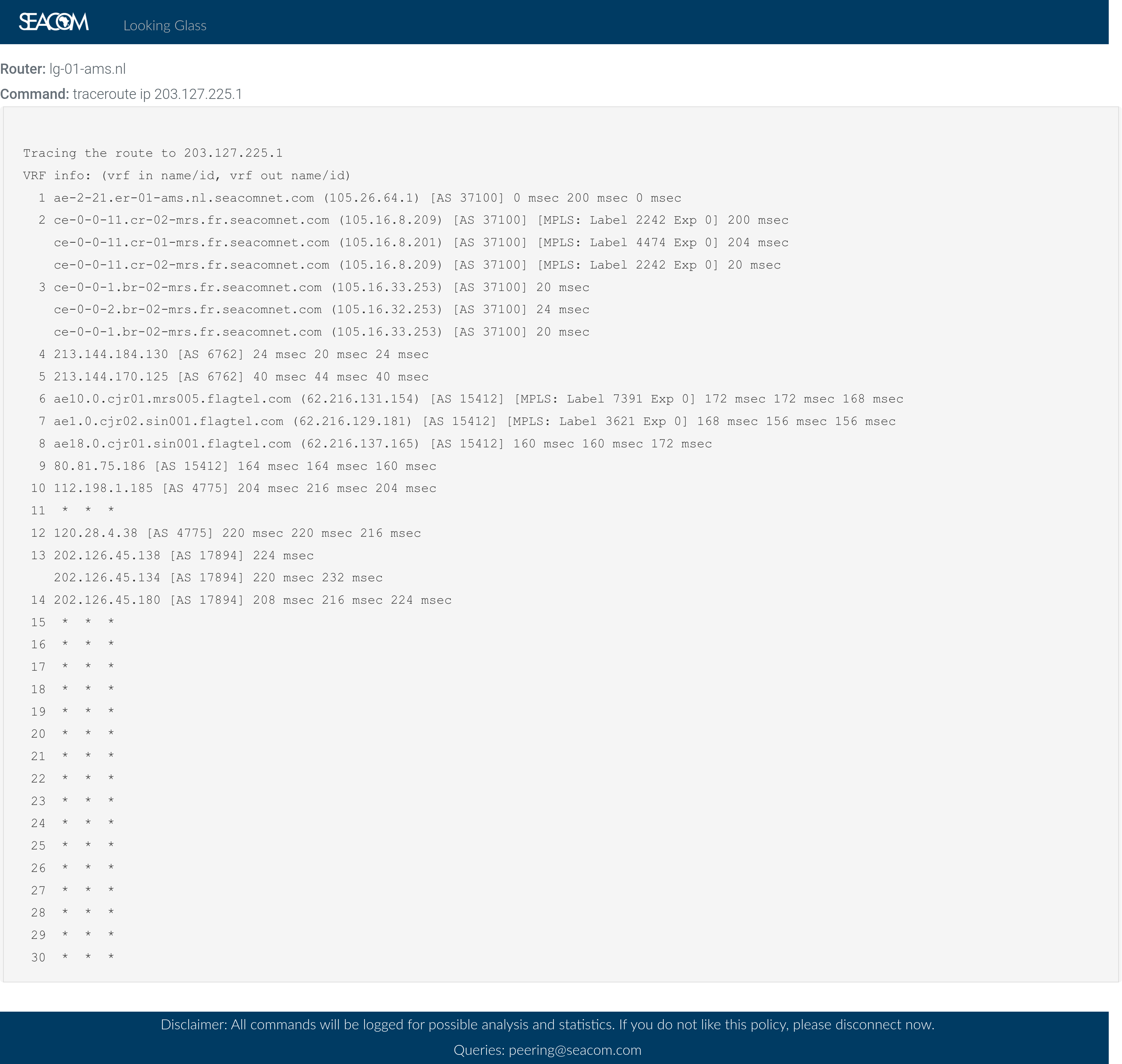}
    \caption{Traceroute from AS37100 to 203.127.225.1.}
    \label{fig:looking-glass-3}
\end{figure}

\section{Topology Compression Method}
\label{sec:appendix:topology-compression}

Given that the route inference complexity grows quadratically with topology size, \ours applies a topology compression method to reduce the topology size after obtaining the Internet topology. The compression is based on the concept of \emph{branch}. Specifically, a branch is defined as a sequence of ASes starting from a stub AS (\ie an AS with no neighbor other than a single provider) and recursively tracing the subject AS's single provider until encountering an AS with more than one provider, more than one customer, or any peers. This AS, leading to the end of the recursion, is denoted as the \emph{access AS} of the branch. The key intuition behind \ours's topology compression strategy is that all routes from branch ASes to non-branch ASes must traverse consecutive C2P links and pass through the corresponding access AS before entering the broader Internet, and vice versa. As a result, the routing tables of branch ASes can be fast computed based on the access AS's routing table, \eg by concatenating a certain slice of the branch. Similarly, the routes towards branch ASes can also be fast computed based on the routes towards the access AS. Besides, routes between ASes within the same branch are trivial. Thus, the routes regarding branch ASes can all be fast computed without involving the branch ASes in the BGP route inference process. Therefore, \ours identifies all branches in the Internet topology, establishes their correspondence with access ASes, and prunes the branches to capture the remaining topology, referred to as the core topology. BGP route inference is performed only on the core topology or its sub-topologies, and routes regarding branch ASes are efficiently computed afterwards based on the inference results. In practice, this compression method reduces the topology size by 36.3\% in terms of vertices and 3.97\% in terms of edges, with a total of 27,775 branches left out.

\section{Non-Branch \textasciitilde PL Computation Validity}
\label{sec:appendix:non-branch-validity}

The two-branch \textasciitilde PL field computation described by Equation~\eqref{eq:pl-update} can be effectively replaced by the branchless computation described by Equation~\eqref{eq:pl-update-non-branch} without affecting the integrity of the results. Here, we prove the validity of this replacement.

\begin{theorem}
Equations~\eqref{eq:pl-update} and \eqref{eq:pl-update-non-branch} produce identical inference results after each iteration.
\end{theorem}
\begin{proof}
The first branch of Equation~\eqref{eq:pl-update}, where $\text{LP}[\hat{p}_{ikj}^{T\textsl{+1}}] \neq (00)_2$, results in $\text{\textasciitilde PL}[\hat{p}_{ikj}^{T\textsl{+1}}] = \text{\textasciitilde PL}[p_{kj}^T]-1$. In comparison, Equation~\eqref{eq:pl-update-non-branch} only differs in the term subtracted, \ie $\text{PL}[p_{ik}^\textsl{1}]$ in Equation~\eqref{eq:pl-update-non-branch} versus $1$ in Equation~\eqref{eq:pl-update}. Thus, we aim to prove $\text{PL}[p_{ik}^\textsl{1}]=1$ under the condition $\text{LP}[\hat{p}_{ikj}^{T\textsl{+1}}] \neq (00)_2$. This is evident, as according to Equation~\eqref{eq:lp-update} and Table~\ref{tab:lp-update-truth-table}, when $\text{LP}[\hat{p}_{ikj}^{T\textsl{+1}}] \neq (00)_2$, it follows that $\text{LP}[p_{ik}^{\textsl{1}}] \neq (00)_2$, indicating an $a_i$-to-$a_k$ relationship exists. In this case, the $a_i$-to-$a_k$ route would be a one-hop, ensuring $\text{PL}[p_{ik}^\textsl{1}]=1$.

The second branch of Equation~\eqref{eq:pl-update}, where $\text{LP}[\hat{p}_{ikj}^{T\textsl{+1}}] = (00)_2$, results in $\hat{p}_{ikj}^{T\textsl{+1}}$ having an overall value of $(00\ 111111)_2$. Since the best route selection subsequently compares $p_{ij}^{T}$ with $\hat{p}_{ikj}^{T\textsl{+1}}$ to choose the higher value, and given that the $p_{ij}^{T}$ is by definition at least $(00\ 111111)_2$, $\hat{p}_{ikj}^{T\textsl{+1}}$ in this case does not affect the result of $p_{ij}^{T\textsl{+1}}$ at the end of this iteration. Similarly, since the output $\text{\textasciitilde PL}[\hat{p}_{ikj}^{T\textsl{+1}}]$ from Equation~\eqref{eq:pl-update-non-branch} will be either $\text{\textasciitilde PL}[p_{kj}^T]-1$ or $\text{\textasciitilde PL}[p_{kj}^T]$, $\hat{p}_{ikj}^{T\textsl{+1}} \le (00\ 111111)_2$ will hold and will not affect the result of $p_{ij}^{T\textsl{+1}}$ either. Thus, Equations~\eqref{eq:pl-update} and \eqref{eq:pl-update-non-branch} produce the same result under the condition $\text{LP}[\hat{p}_{ikj}^{T\textsl{+1}}] = (00)_2$ as well.
\end{proof}

\section{Simplifying Priority Byte Selection}
\label{sec:appendix:best-route-selection-simplification}

We can effectively eliminate $p_{ij}^{T}$ from the comparison among route priority bytes in each iteration, thereby simplifying Equation~\eqref{eq:byte-selection} to \eqref{eq:byte-selection-simplified}. The rationale behind this simplification is that there is always some $k$ such that $\hat{p}_{ikj}^{T\textsl{+1}} \geq p_{ij}^T$. Here, we provide the proof.

\begin{lemma}\label{lemma:byte-selection-simplification-1}
There exist a certain $T'$ ($0\le T' \le T$) and a certain $k$ ($k \in \{0,\dots,n-1\}$) such that $p_{ij}^T=\hat{p}_{ikj}^{T'}$.
\end{lemma}
\begin{proof}
According to Equation~\eqref{eq:byte-selection}, $p_{ij}^T$ either equals $max\{ \hat{p}_{ikj}^{T} \}_{k=0}^{n\textsl{-1}}$ or $p_{ij}^{T\textsl{-1}}$. If $p_{ij}^T=max\{ \hat{p}_{ikj}^{T} \}_{k=0}^{n\textsl{-1}}$, apparently there exists a certain $k$ when $T'=T$ such that $p_{ij}^T=\hat{p}_{ikj}^{T'}=max\{ \hat{p}_{ikj}^{T} \}_{k=0}^{n\textsl{-1}}$. If $p_{ij}^T=p_{ij}^{T\textsl{-1}}$, the problem reduces to finding a certain $T'$ ($0\le T' \le T-1$) and a certain $k$ ($k \in \{0,\dots,n-1\}$) such that $p_{ij}^{T\textsl{-1}}=\hat{p}_{ikj}^{T'}$. We repeat this reduction recursively until the problem reduces to finding a certain $T'$ ($0\le T' \le 1$) and a certain $k$ ($k \in \{0,\dots,n-1\}$) such that $p_{ij}^{\textsl{1}}=\hat{p}_{ikj}^{T'}$. Let $T'=1$ and $k=j$, and then we need prove $p_{ij}^{\textsl{1}}=\hat{p}_{ijj}^{\textsl{1}}=max\{ \hat{p}_{ikj}^{1} \}_{k=0}^{n\textsl{-1}}$. Note that $\hat{p}_{ikj}^{\textsl{1}}$ is computed based on $p_{ik}^{\textsl{1}}$ and $p_{kj}^{\textsl{0}}$. By definition, $p_{kj}^{\textsl{0}}=(11\ 111111)_2$ when $k=j$, otherwise $p_{kj}^{\textsl{0}}=(00\ 111111)_2$. Then, according to the definition of $p_{ij}^{\textsl{1}}$ (see Equation~\eqref{eq:L-definition}) and the computation of $\hat{p}_{ikj}^{1}$ (see Equations~\eqref{eq:lp-update} to \eqref{eq:pl-update-non-branch}), it follows that $p_{ij}^{\textsl{1}} = \hat{p}_{ijj}^{\textsl{1}}\ge (00\ 111111)$, while all $\hat{p}_{ikj}^{\textsl{1}}=(00\ 111111)$ when $k\neq j$. Thus, we have proved the reduced problem, and consequently, the original problem is also proved.
\end{proof}

\begin{lemma}\label{lemma:byte-selection-simplification-2}
$\hat{p}_{ikj}^{T'} \le \hat{p}_{ikj}^{T}$ when $T' < T$.
\end{lemma}
\begin{proof}
$\hat{p}_{ikj}^{T}$ is computed based on $p_{ik}^{\textsl{1}}$ and $p_{kj}^{T}$. According to Table~\ref{tab:lp-update-truth-table}, when $p_{ik}^{\textsl{1}}$ is fixed, the LP field of $\hat{p}_{ikj}^{T}$ either increases or remains unchanged as $p_{kj}^T$ increases. Since, according to Equation~\eqref{eq:byte-selection}, $p_{kj}^{T} \ge p_{kj}^{T'}$ when $T > T'$, it follows that $\text{LP}[\hat{p}_{ikj}^{T}] \ge \text{LP}[\hat{p}_{ikj}^{T'}]$ when $T > T'$. If $\text{LP}[\hat{p}_{ikj}^{T}] > \text{LP}[\hat{p}_{ikj}^{T'}]$, it is apparent that $\hat{p}_{ikj}^{T'} \le \hat{p}_{ikj}^{T}$, as the LP field represents the two most significant bits of the priority byte. If $\text{LP}[\hat{p}_{ikj}^{T}] = \text{LP}[\hat{p}_{ikj}^{T'}]$, there are two cases to consider: First, if $p_{kj}^{T} = p_{kj}^{T'}$, then $p_{ikj}^{T} = p_{ikj}^{T'}$. Second, if $p_{kj}^{T} > p_{kj}^{T'}$ but both $\text{LP}[p_{ik}^\textsl{1}] \odot \text{LP}[p_{kj}^T]$ and $\text{LP}[p_{ik}^\textsl{1}] \odot \text{LP}[p_{kj}^{T'}]$ result in $(00)_2$, then according to Equation~\eqref{eq:pl-update}, both $p_{ikj}^{T'}$ and $p_{ikj}^{T}$ equal $(00\ 111111)_2$. Therefore, $p_{ikj}^{T} = p_{ikj}^{T'}$ is ensured when $\text{LP}[\hat{p}_{ikj}^{T}] = \text{LP}[\hat{p}_{ikj}^{T'}]$. In all cases, $\hat{p}_{ikj}^{T'} \le \hat{p}_{ikj}^{T}$ when $T' < T$.
\end{proof}

\begin{theorem}
There exists a certain $k$ such that $\hat{p}_{ikj}^{T\textsl{+1}} \ge p_{ij}^T$.
\end{theorem}
\begin{proof}
According to Lemma~\ref{lemma:byte-selection-simplification-1}, there exist a certain $T'$ ($0\le T' \le T$) and a certain $k$ ($k \in \{0,\dots,n-1\}$) such that $p_{ij}^T=\hat{p}_{ikj}^{T'}$. Then, according to Lemma~\ref{lemma:byte-selection-simplification-2}, $\hat{p}_{ikj}^{T'} \le \hat{p}_{ikj}^{T+1}$ since $T' < T+1$. So $\hat{p}_{ikj}^{T\textsl{+1}} \ge p_{ij}^T$ with this $k$.
\end{proof}

\section{Restoring Routes from Next-Hop Matrix}
\label{sec:appendix:restoring-route-from-next-hop}

Here, we describe the process of restoring the complete path of a route between two ASes using the state matrix $P^T$ and the next-hop matrix $N^T$. Algorithm~\ref{alg:restore-path} outlines the steps to achieve this. The function \textproc{RestorePath} takes as inputs the vantage point AS $a_i$, the origin AS $a_j$, the state matrix $P^T$, and the next-hop matrix $N^T$. It initializes the next-hop index with $i$ (see Line 2) and an empty list to store the path (see Line 3). The algorithm proceeds in a loop, continuously updating the next-hop AS until it reaches the origin AS $a_j$. If at any point the state matrix indicates that the origin is unreachable (see Line 5), the function returns \texttt{None}. Otherwise, it updates the next-hop using the next-hop matrix $N^T$ (see Line 7) and appends the corresponding AS number to the path (see Line 8). The process repeats until the origin AS is reached, at which point the complete path is returned.

Algorithm~\ref{alg:restore-path} ensures that we can restore the AS path for any given $(i,j)$ pair based on the information stored in the state and next-hop matrices, providing a systematic approach to generate the complete route set from the matrix-based results.

\begin{algorithm}[ht]
    \caption{Restoring Path from Next-Hop Matrix}
    \label{alg:restore-path}
    \begin{spacing}{1.2}
    \begin{algorithmic}[1]
        \Function{RestorePath}{$a_i$, $a_j$, $P^T$, $N^T$}
            \State $next\_hop = i$
            \State $path = [\ ]$
            \While{$next\_hop \neq j$}
                \If{$P^T[next\_hop,\ j] \leq 0\text{b}00111111$}
                    \State \Return None
                \EndIf
                \State $next\_hop = N^T[next\_hop,\ j]$
                \State $\textproc{Append}(path,\ a_{next\_hop})$
            \EndWhile
            \State \Return $path$
        \EndFunction
    \end{algorithmic}
    \end{spacing}
\end{algorithm}

\section{Community Engagement}
We preliminarily shared our work with APNIC experts, who raised concerns about the pratical significance of real-world incidents. In response to this feedback, we carried out the empirical study presented in Section~\ref{sec:empirical-study}, providing quantitative evidence of the prevalence and impact of stealthy hijacking in today’s Internet. Additionally, we are actively promoting an informational Internet Draft\footnote{https://datatracker.ietf.org/doc/draft-li-sidrops-stealthy-hijacking/} that formalizes the mechanism and properties of stealthy hijacking and aims to raise community awareness of this threat.

\section{Additional Technical Reports}
\label{sec:appendix:additional-reports}

We provide extended technical details covering algorithms, implementation notes, and in-depth analyses in standalone technical reports available via external links. Readers interested in these topics may refer to:
\begin{itemize}
    \item Report on \emph{\href{https://github.com/yhchen-tsinghua/stealthy-bgp-hijacking/blob/main/docs/stealthy-hijacking-discovery-as-a-service.pdf}{Stealthy Hijacking Discovery as a Service}}
    \item Report on \emph{\href{https://github.com/yhchen-tsinghua/stealthy-bgp-hijacking/blob/main/docs/matrix-based-route-priority-update.pdf}{Matrix-Based Route Priority Update}}
    \item Report on \emph{\href{https://github.com/yhchen-tsinghua/stealthy-bgp-hijacking/blob/main/docs/time-space-tradeoff-in-matrix-update.pdf}{Time-Space Tradeoff in Matrix Update}}
\end{itemize}

\clearpage

\section{Artifact Appendix}

\subsection{Description \& Requirements}

This artifact supports the paper \textit{Understanding the Stealthy BGP Hijacking Risk in the ROV Era}, and enables full reproduction of all experiments and results presented therein. The artifact includes three parts: (i) an empirical study based on real-world BGP data (presented in \S\ref{sec:empirical-study} of the paper), (ii) an analytical study through matrix-based BGP route inference (presented in \S\ref{sec:risk-assessment} of the paper), and (iii) a performance evaluation of the proposed analytical framework (presented in \S\ref{sec:performance-evaluation} of the paper). Additionally, the implementation of our matrix-based BGP route inference algorithm is encapsulated in a Python package named \texttt{matrix-bgpsim}\footnote{Stable versions of the package are officially released on \url{https://pypi.org/project/matrix-bgpsim/} and can be installed via \texttt{pip}.}, and our service in production is available at \url{https://yhchen.cn/stealthy-bgp-hijacking}.

Overall, the artifact contains following components:
\begin{itemize}
    \item Source code and scripts for each part of the experiment.
    \item Pre-processed datasets (\eg inferred matrices).
    \item Pre-computed results and cache files for boosting long-running steps.
    \item Docker and Conda configurations for environment setup.
    \item Documentation and README for usage instructions.
\end{itemize}

To facilitate reproduction and reduce platform dependency, we provide a fully-configured cloud-based evaluation platform (access credentials available via HotCRP). Alternatively, users can run the artifact locally via Docker or by manually setting up the environment with the README instructions.

\noindent
\textbf{How to access.} DOI \href{https://doi.org/10.5281/zenodo.16565359}{10.5281/zenodo.16565359} or \href{https://github.com/yhchen-tsinghua/stealthy-bgp-hijacking}{GitHub}.

\noindent
\textbf{Hardware dependencies.} The artifact requires the following resources to complete all experiments:

\begin{itemize}
    \item At least 120~GB of system free memory,
    \item At least 60~GB of available disk space,
    \item Internet access for data downloads, and
    \item Nvidia GPU and CUDA support for full benchmarking.
\end{itemize}

\noindent
\textbf{Software dependencies.} The artifact provides two setups:
\begin{itemize}
    \item Docker-based: A pre-built Docker image is provided. The user only needs the Docker tool suite installed.
    \item Manual setup: Detailed scripts are included to set up required packages and tools. These include:
    \begin{itemize}
        \item Conda environment with Python 3.10 and dependencies listed in ``environment.yml'',
        \item Node.js for the frontend display, and
        \item C/C++ toolchains for building third-party tools, including
            \texttt{bgpdump} and \texttt{bgpsim}.
    \end{itemize}
\end{itemize}

\noindent
\textbf{Benchmarks.} The artifact relies on several data sources and benchmark components:
\begin{itemize}
    \item BGP routing data: RIB snapshots from RouteViews collectors \emph{wide}, \emph{amsix}, and \emph{route-views2} since January 1, 2025, and discovered incidents from a two-month period (January to February 2025) of our deployed service.
    \item Intermediate matrices: Pre-computed Internet-scale BGP route matrices, derived from CAIDA AS relationship datasets as of 2025/01/01. These are required for the analytical and performance studies.
    \item Synthetic topologies: Generated using a sampled Internet AS graph (10,000 ASes) for benchmarking BGP route inference performance.
\end{itemize}

\subsection{Artifact Installation \& Configuration}

The artifact supports three installation modes:

\noindent
\textbf{Cloud Platform.} A fully pre-configured cloud environment is provided for evaluation purposes. Users can simply SSH into the platform (credentials provided via HotCRP) and launch the evaluation environment using a one-line startup script. No additional setup is required.

\noindent
\textbf{Docker-based.} The artifact includes Docker Compose configurations. Users only need to install the Docker tool suite and download the provided container image and dataset archive. BASH scripts are provided to launch the evaluation environment with all components properly mounted and configured.

\noindent
\textbf{Manual Setup.} This involves installing required system packages, Python dependencies (via Conda), and third-party tools such as \texttt{bgpdump} and \texttt{bgpsim}. Detailed instructions and setup scripts are included in the repository.

\subsection{Experiment Workflow}

The experiment has three main parts that correspond to the core sections of the paper: an empirical study (\S\ref{sec:empirical-study}), an analytical study (\S\ref{sec:risk-assessment}), and performance evaluation (\S\ref{sec:performance-evaluation}). Each part is self-contained and run via a corresponding ``run.sh''

\noindent
\textbf{Empirical Study.} This part runs a backend routine to analyze real-world BGP routing data and discover stealthy hijacking incidents. Then, it characterizes the discovered incidents and reproduces all relevant figures and tables in \S\ref{sec:empirical-study}. Finally, it sets up a frontend service locally for interactive display of the discovered incidents. A production version of this service is also available at \url{https://yhchen.cn/stealthy-bgp-hijacking}.

\noindent
\textbf{Analytical Study.} This part first uses a matrix-based approach to infer BGP routes and quantify the stealthy hijacking risk under partial ROV deployment. It stores the analysis results in intermediate matrices on the disk, and then reproduces all figures and tables in \S\ref{sec:risk-assessment}.

\noindent
\textbf{Performance Evaluation.} This part benchmarks the runtime of our approach against existing tools on synthetic Internet topologies. It then validates the analytical results against the empirical results to evaluate the accuracy of our analytical framework. It further evaluates the accuracy of our analytical framework under ablation of input data sources and statistically analyzes how robust our analytical framework is against mislabels in input data sources. It finally reproduces all figures and tables in \S\ref{sec:performance-evaluation}.

\subsection{Major Claims}

The major claims supported by our artifact are as follows:
\begin{itemize}
    \item (C1) Stealthy hijacking in the wild is mostly short-lived and targets sub-prefixes, with new cases emerging almost daily and some persisting long-term. Its exposure is sensitive to vantage point selection.
    \item (C2) The current partial ROV deployment significantly amplifies stealthy hijacking risk from 0 to a 14.1\% overall success probability.
    \item (C3) Targeted stealthy hijacking achieves near-certain success on specific AS pairs (up to 99.5\%).
    \item (C4) Stealthy hijacking risk mostly opposes the overall risk trend across ASes but is eventually suppressed as ROV’s restrictions on attackers prevail.
    \item (C5) ASes most effective in launching stealthy hijacking are in Europe, South America, and North America.
    \item (C6) Cumulative AS hegemony shows the strongest quadratic correlation with stealthy hijacking risk.
    \item (C7) A small fraction of risk-critical and ROV-enabled ASes account for the majority of stealthy hijacking risk.
    \item (C8) Validation on real-world datasets shows up to 95.9\% incident-level accuracy of our analytical framework.
    \item (C9) Integrating multiple reliable sources to obtain a more complete view of ROV deployment is crucial to achieve accurate stealthy hijacking risk assessment.
    \item (C10) Our matrix-based route inference achieves a 500-fold speedup against existing baseline methods.
\end{itemize}

\subsection{Evaluation}

\noindent
\textbf{Experiment (E1)}
[Empirical Study]: discover stealthy hijacking incidents with real-world routing data, reproduce all figures and tables in \S\ref{sec:empirical-study} to support C1, and additionally set up a frontend service to display discovered incidents interactively.

\noindent
[Preparation]
No extra preparation is needed.

\noindent
[Execution]
Running ``empirical-study/run.sh'' is all you need. This script will execute the following steps:
\begin{enumerate}
    \item Run the backend routine that discovers stealthy BGP hijacking incidents using RouteViews RIBs from collectors \emph{wide}, \emph{amsix}, and \emph{route-views2}, each captured at 12:00 January 1, 2025. This is a scaled-down demo for one-day discovery. In actual deployment, we register a cron-job to call this backend routine daily.
    \item Reproduce all figures and tables in \S\ref{sec:empirical-study} to support C1. It uses all incidents and alarms captured in the first two months of year 2025, which are preserved in the artifact beforehand and are exactly a snapshot of the results by 2025/07/11 from our service in production at \url{https://yhchen.cn/stealthy-bgp-hijacking}.
    \item Set up a frontend service to display discovered incidents.
\end{enumerate}

\noindent
[Results]
The discovered incidents will be saved to ``empirical-study/results'', including one JSON file for alarms and one JSON file for incidents. Figure~\ref{fig:incidents-breakdown}-\ref{fig:vp-distribution} in PDF format and Table~\ref{tab:incident-impact} in JSON format are reproduced under the same result directory, which support C1. Once the frontend service is up, it can be accessed at \url{http://localhost:3000/} using a browser.

\vspace{2mm}
\noindent
\textbf{Experiment (E2)}
[Analytical Study]: perform matrix-based BGP route inference using CAIDA AS relationship data, analyze the stealthy BGP hijacking risk, and reproduce all figures and tables in \S\ref{sec:risk-assessment} to support C2-C7.

\noindent
[Preparation]
No extra preparation is needed.

\noindent
[Execution]
Running ``analytical-study/run.sh'' is all you need. This script will execute the following steps:
\begin{enumerate}
    \item Run a risk analysis process that infers complete BGP routes on the benign reach and the malicious reach. It then characterizes stealthy hijacking risk based on these inferred routes and stores the results in intermediate matrices on the disk. It further computes various topological features on each AS and cache the results on the disk.
    \item Reproduce all figures and tables in \S\ref{sec:risk-assessment} using the intermediate matrices and AS features generated in the previous step.
\end{enumerate}

\noindent
[Results]
The intermediate results including matrices and AS features will be saved under ``analytical-study/data/matrices/'' in the format of compressed pickle objects and CSV files. A JSON file including statistics of routes, Figure~\ref{fig:rov-measurements}-\ref{fig:risk-attribution} in PDF format, and Table~\ref{tab:risk-dissection} in Tex format are reproduced under ``analytical-study/results'', which support C2-C7.

\vspace{2mm}
\noindent
\textbf{Experiment (E3)}
[Performance Evaluation]: benchmark the runtime of \texttt{matrix-bgpsim} and the baseline \texttt{bgpsim}, evaluate accuracy and robustness of our analytical framework, and reproduce all figures and tables in \S\ref{sec:performance-evaluation} to support C8-C10.

\noindent
[Preparation]
No extra preparation is needed.

\noindent
[Execution]
Running ``performance-evaluation/run.sh'' is all you need. This script will execute the following steps:
\begin{enumerate}
    \item Benchmark the runtime of our \texttt{matrix-bgpsim} and the baseline \texttt{bgpsim}. This first generates a sampled Internet topology that contains 10,000 ASes, based on CAIDA serial-2 AS relationship dataset on 2025/01/01. This process starts with Tier-1 mesh and progressively adds new ASes that connects to the existing topology, until the number of ASes reaches 10,000. Then, it tests how long \texttt{matrix-bgpsim} and \texttt{bgpsim} generates all routes between any random 10 to 100 ASes, respectively, on the aforementioned sampled topology. \texttt{matrix-bgpsim} is also tested under varying number of CPU processes and GPU. 
    \item Use the results from of the previous step and the discovered real-world incidents to evaluate the performance of our framework in terms of accuracy, robustness, and efficiency, and reproduce all figures and tables in \S\ref{sec:performance-evaluation}.
\end{enumerate}

\noindent
[Results]
The benchmark results are stored in CSV format under ``performance-evaluation/.cache''. Figure~\ref{fig:accuracy-threshold}-\ref{fig:runtime-performance} in PDF format and Table~\ref{tab:rov-ablation} in Tex format are created under ``performance-evaluation/results'', which support C8-C10.

\end{document}